%% file: Errormitigation_arxivposting.tex
\documentclass{ieeetran}
\usepackage{cite}
\usepackage{subfigure}
\usepackage{amsmath,amssymb,amsfonts}
\usepackage{blindtext}
\usepackage{graphicx}
\usepackage{textcomp}
\usepackage{cite}
\usepackage{amsmath,amssymb,amsfonts}
\usepackage{comment}
\usepackage{caption}
\usepackage{graphicx}
\usepackage{textcomp}
\usepackage{xcolor}
\input{macros-ieee.tex}

\def\BibTeX{{\rm B\kern-.05em{\sc i\kern-.025em b}\kern-.08em
    T\kern-.1667em\lower.7ex\hbox{E}\kern-.125emX}}
\usepackage{algorithm,algcompatible,amsmath}
\algnewcommand\INPUT{\item[\textbf{Input:}]}%
\algnewcommand\OUTPUT{\item[\textbf{Output:}]}%
\def \thetabf{{\theta}}

\def \Circ{\mathrm{circ}}
\def \xbf {\mathbf{x}}
\def \var {\mathrm{var}}

\def \QEM {\mathrm{QEM}}
\def \Tr {\mathrm{Tr}}
\def \sgn {\mathrm{sgn}}

\def \sbar{s}
\def \sbfbar{s^D}

\def \Tr{\mathrm{Tr}}

\def\BibTeX{{\rm B\kern-.05em{\sc i\kern-.025em b}\kern-.08em
    T\kern-.1667em\lower.7ex\hbox{E}\kern-.125emX}}
\begin{document}

\title{Error Mitigation-Aided Optimization of Parameterized Quantum Circuits: Convergence Analysis}
\author{\IEEEauthorblockN { Sharu Theresa Jose,  Osvaldo Simeone}
\thanks{ST Jose (sharutheresa@gmail.com) is with the Department of Computer Science, University of Birmingham and  OS (osvaldo.simeone@kcl.ac.uk) is with the Department of Engineering of King's College London. This work was partly done when ST Jose was a postdoctoral researcher at King's College London. The  work has been funded by  the European Research Council (ERC) under the European Union’s Horizon 2020 Research and Innovation Programme (Grant Agreement No. 725731).
}
}
\maketitle
\begin{abstract}
 Variational quantum algorithms (VQAs) offer the most promising path to obtaining quantum advantages via noisy intermediate-scale  quantum (NISQ) processors. Such systems leverage classical optimization to tune the parameters of a  parameterized quantum circuit (PQC). The goal is minimizing a cost function that depends on measurement outputs obtained from the PQC. Optimization is typically implemented via stochastic gradient descent (SGD). On NISQ computers, gate noise due to imperfections and decoherence affects the stochastic gradient estimates by introducing a bias. Quantum error mitigation (QEM) techniques can reduce the estimation bias without requiring any increase in the number of qubits, but they in turn cause an increase in the variance of the gradient estimates. This work studies the impact of quantum gate noise on the convergence of SGD for the variational eigensolver (VQE), a fundamental instance of VQAs. The main goal is ascertaining conditions under which QEM can enhance the performance of SGD for VQEs. It is shown that quantum gate noise induces a non-zero error-floor on the convergence error of SGD (evaluated with respect to a reference noiseless PQC), which depends on the number of noisy gates, the strength of the noise, as well as the eigenspectrum of the observable being measured and minimized. In contrast, with QEM, any arbitrarily small error can be obtained. Furthermore, for error levels attainable with or without QEM, QEM can reduce the number of required iterations, but only as long as the quantum noise level is sufficiently small, and a sufficiently large number of measurements is allowed at each SGD iteration. Numerical examples for a max-cut problem corroborate the main theoretical findings.

\end{abstract}
\section{Introduction}
\label{sec:introduction}

\subsection{Motivation}
Current noisy intermediate-scale quantum (NISQ)  processors are severely limited by the availability of a small number of qubits and by the unavoidable errors due to \textit{quantum gate noise}. Quantum gate noise refers to unwanted interactions of qubits with the environment, leading to \textit{decoherence}; as well as to the  imperfect execution of user-specified quantum gates on the physical device, leading to \textit{gate infidelity}.  While fault-tolerant computation is not feasible on NISQ processors due to the small number of available qubits \cite{aharonov2008fault}, forms of \textit{quantum error mitigation} (QEM), which do not require the addition of further qubits, are possible  \cite{temme2017error,endo2018practical}. Many promising applications of NISQ devices involve  \textit{parameterized quantum circuits (PQCs)}, whose architecture -- or \textit{ansatz} -- is  compatible with  NISQ processors. As illustrated in Fig.~\ref{fig:VQA}, variational quantum algorithms (VQAs) tune the parameters $\theta$ defining the operation of a PQC via a classical optimizer that relies on measurements from the PQC \cite{schuld2015introduction, simeone2022introduction}. This paper studies the impact of quantum gate noise on the performance of VQAs by focusing on the question of whether QEM can be effective in enhancing the performance of an optimized PQC and/or in reducing the optimization complexity.


\begin{figure*}
    \centering
    \includegraphics[scale=0.46,clip=true, trim= 3.5in 3.1in 2.6in 0.9in]{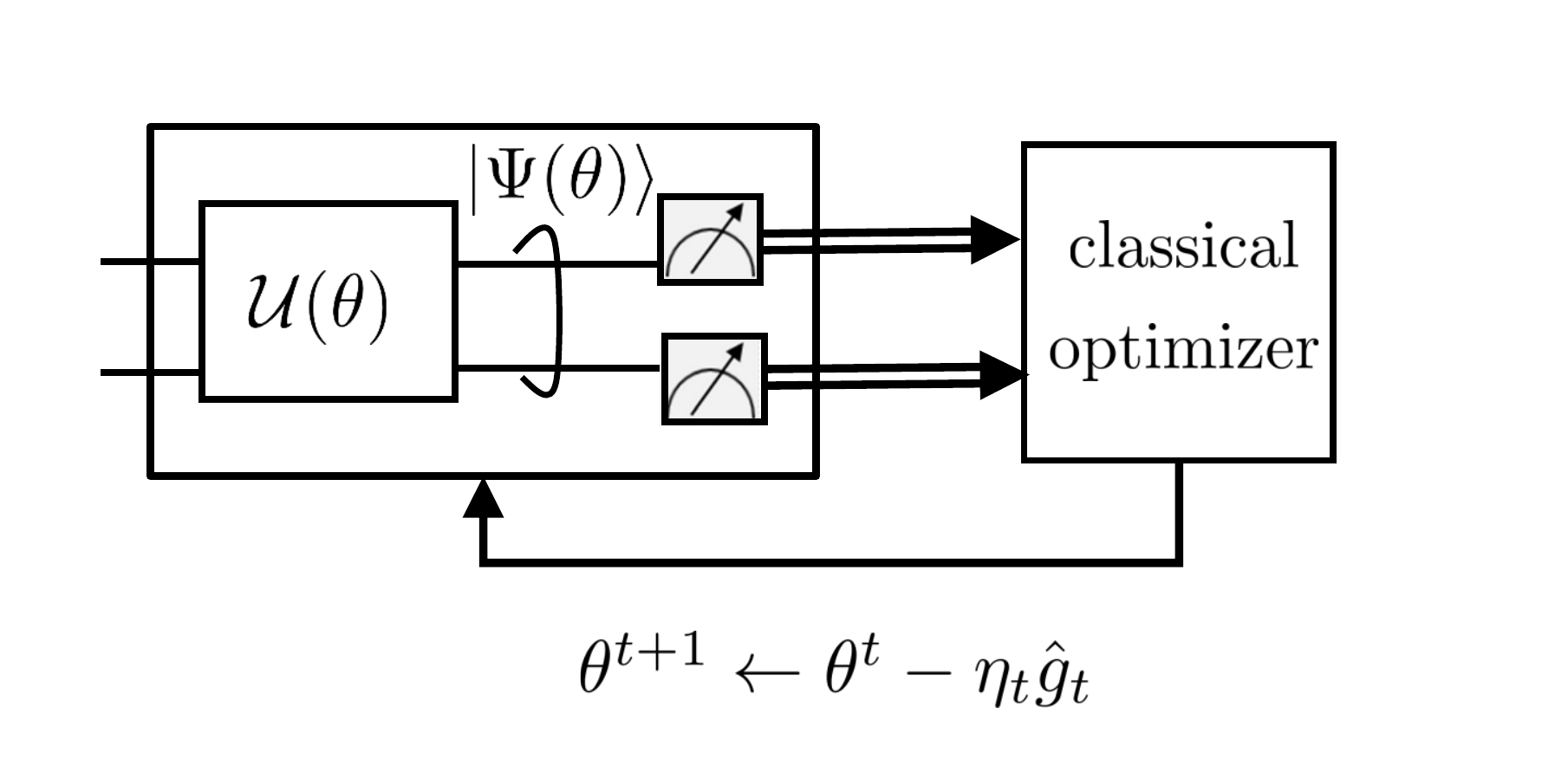}
    \caption{Illustration of the variational quantum algorithm (VQA) framework: The model parameter vector $\thetabf$, specifying the operation of a, generally noisy, parameterized quantum circuit (PQC) $\Uscr(\thetabf)$, is optimized via stochastic gradient descent (SGD) based on measurement outputs. The measurement outputs are used to produce an estimate $\hat{g}_t$ of the gradient. The estimate $\hat{g}_t$ is affected by the  inherent randomness of quantum measurements, as well as by the noise introduced by the imperfections of the quantum gates within the PQC.}\label{fig:VQA}
    \centering
    \includegraphics[scale=0.52,clip=true, trim= 0in 0in 0in 0in]{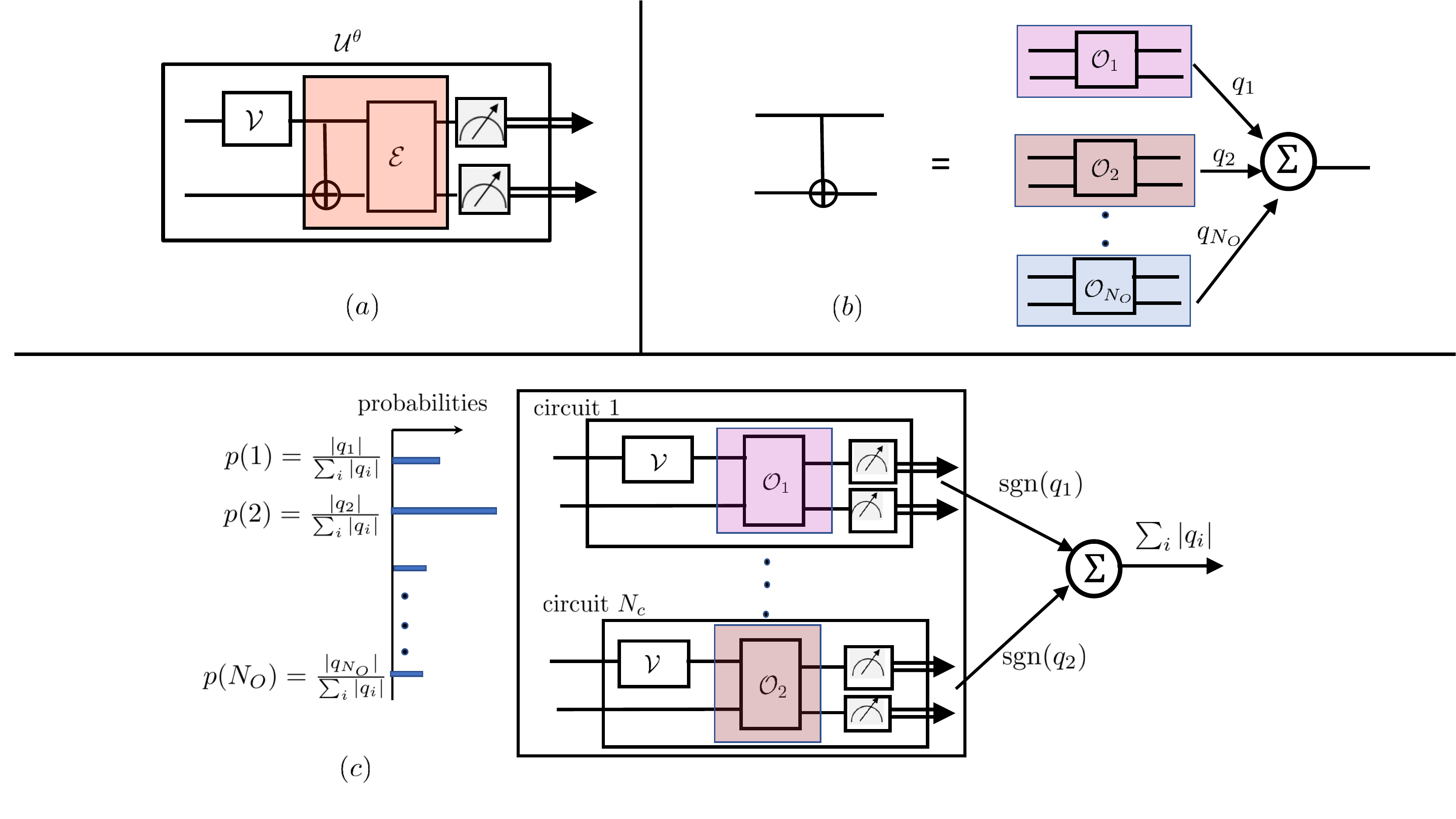}
    \caption{A VQA aided by quasiprobabilistic quantum error mitigation (QEM): $(a)$ a noisy PQC  where gate noise channel $\Escr$ acts on a CNOT gate; $(b)$ quasi-probabilistic representation of a CNOT gate as a linear combination of a set of operations $\{\Oscr_i\}$ implementable on the quantum computer; $(c)$ QEM approach, whereby multiple circuits are sampled, and their outputs combined, in order to approximate the operation of the  PQC in the absence of quantum gate noise.  }
    \label{fig:QEM}
\end{figure*}

To address this question, we consider the most fundamental instance of VQAs, namely the \textit{variational quantum eigensolver (VQE)} \cite{peruzzo2014variational}. The goal of the VQE is to ensure that the output state of the PQC, denoted as $|\Psi(\thetabf)\rangle$, provides a good approximation of the ground state of a given observable $H$. Mathematically, we wish to approximately solve the problem  \begin{equation}
    \thetabf^{*} = \arg\min_{\thetabf} \langle H \rangle _{\vert \Psi(\thetabf) \rangle} \label{eq:VQE}
\end{equation} of minimizing  the expected value of observable $H$ when evaluated for the output state $|\Psi(\thetabf)\rangle$. Assuming that the ansatz of the PQC is sufficiently expressive,  the solution $|\Psi(\thetabf^*)\rangle$ obtained from problem \eqref{eq:VQE} is a close approximation of the desired eigenstate. Applications of VQEs include the solution of quadratic unconstrained binary optimization (QUBO) \cite{garcia2018addressing,amaro2022filtering}, as well as problems in quantum chemistry \cite{moll2018quantum}. 

In a VQE, the optimization (\ref{eq:VQE}) of the parameters $\thetabf$ of the PQC is typically carried out in an iterative manner by means of \textit{stochastic gradient descent (SGD)}. SGD estimates the gradient $\nabla_{\thetabf} \langle H \rangle_{|\Psi(\thetabf)\rangle}$ of the expected value $\langle H \rangle_{\vert \Psi(\thetabf) \rangle}$ via measurements from the output of the PQC  \cite{schuld2015introduction,simeone2022introduction}. Such estimates are affected by the inherent randomness of quantum measurements, as well as by the noise introduced by the imperfections  of the quantum gates within the PQC and by decoherence. As a result of such quantum gate noise, as illustrated in Fig. \ref{fig:effectofQEM}-(a), the estimates of the gradient $\nabla_{\thetabf} \langle H \rangle_{|\Psi(\thetabf)\rangle}$ are \textit{biased} \cite{gentini2020noise}.

 QEM techniques provide an \textit{algorithmic} approach to mitigate quantum gate noise in NISQ devices. Unlike quantum error correction codes \cite{shor1995scheme, steane1996error,aharonov2008fault}, quantum error mitigation requires no additional qubit resources. Examples of quantum error mitigation techniques include \textit{quasiprobabilistic QEM}  \cite{temme2017error}, zero noise extrapolation \cite{temme2017error,li2017efficient},  randomized compiling \cite{wallman2016noise}, and Pauli-frame randomization \cite{knill2005quantum}. In this work, we will focus on quasiprobabilistic QEM, which will be referred as QEM for short. 
 
 As illustrated in Fig. \ref{fig:QEM}, QEM provides protection against quantum gate noise by running multiple, $N_c$, noisy quantum circuits, which are sampled from a set of circuits implementable on the NISQ device, and by post-processing the measurement outputs. As sketched in Fig.~\ref{fig:effectofQEM}-(b),  QEM can reduce the mentioned quantum noise-induced bias in the measurement of a quantum observable, while generally increasing the corresponding variance \cite{endo2021hybrid}. 

Using QEM, one can hence obtain a \textit{less biased} estimate of the gradient $\nabla_{\thetabf} \langle H \rangle_{|\Psi(\thetabf)\rangle}$. This can potentially improve the convergence of SGD when applied to the VQE problem \eqref{eq:VQE}. However, the bias reduction should be weighted against the variance increase as a function of the number of quantum measurements that one can afford at any SGD iteration. Under which conditions is it advantageous to trade an increased variance for a decrease in the bias in VQEs?

 \begin{figure}
     \centering
      \includegraphics[scale=0.5,clip=true, trim= 3.8in 1.2in 4.5in 0.8in]{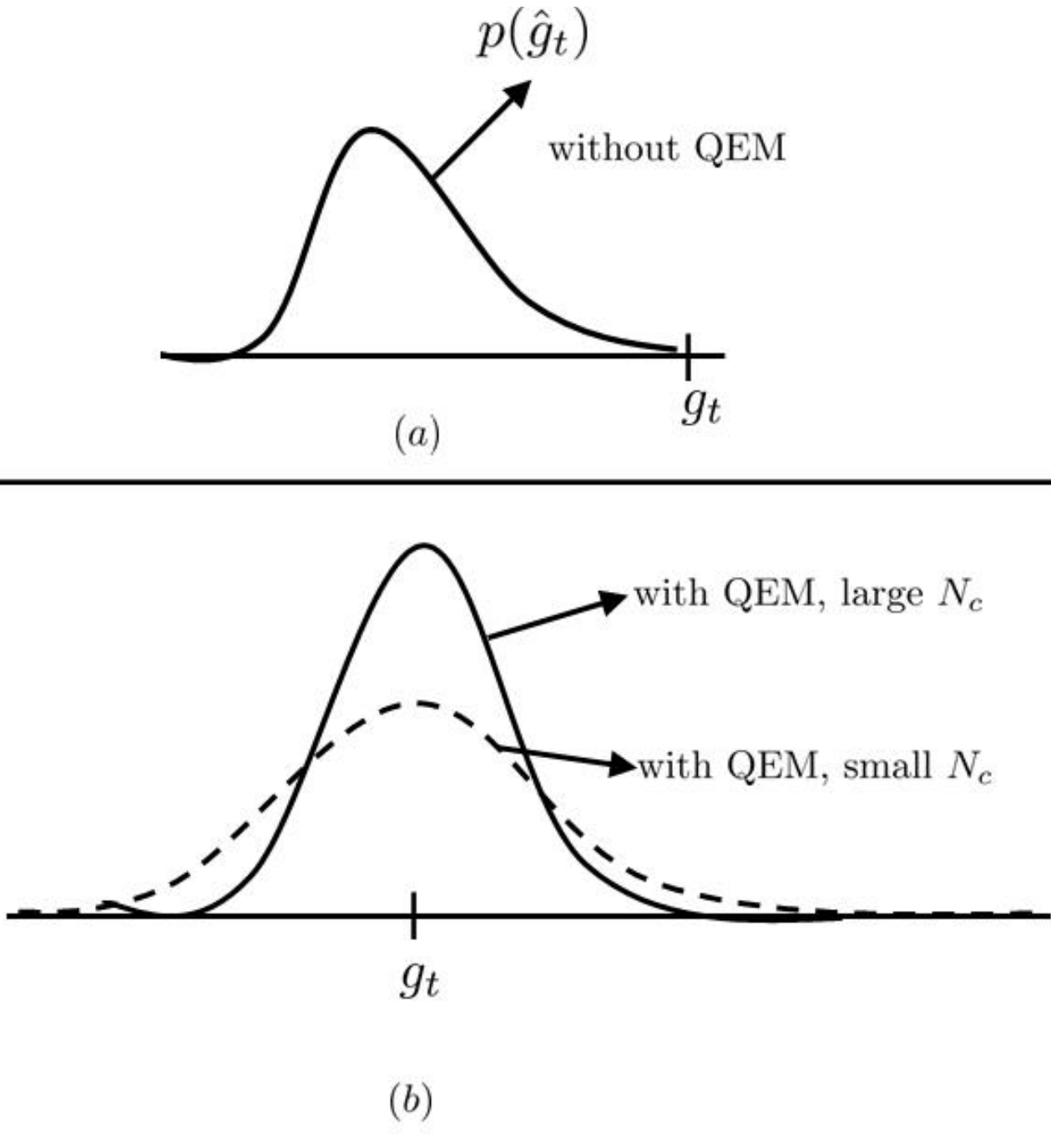}
     \caption{Illustration of the effect of QEM: (top) without QEM, the presence of quantum gate noise results in a biased gradient estimator; (bottom) employing QEM by sampling $N_c$ circuits can decrease the bias, at the cost of increase in the variance, unless $N_c$ is sufficiently large.
     }
     \label{fig:effectofQEM}
 \end{figure}
 \subsection{Main Contributions}
 \begin{figure}
 \begin{center}
 \begin{tabular}{|c|c|}
 \hline
  \textbf{ Description}  &  \textbf{Notation}\\
  \hline 
  Number of parameters of PQC & $D$\\ \hline 
  Number of measurements per iteration & $2DN_m$\\\hline 
  Number of sampled noisy circuits & $N_c$ \\ \hline 
  Noise level & $\gamma \in [0,1]$\\
  \hline 
 \end{tabular}
 \captionof{table}{Main notations used in the paper}
 \end{center}
 \end{figure}
\begin{figure*}[h!]
\begin{minipage}[c]{0.35 \linewidth}
     \centering
\begin{tabular}{|c|c|}
\hline
  \textbf{ Schemes}  &  \textbf{Iteration Complexity} \\
     \hline 
  shot-noise   only   & $\widetilde{\Oscr} \Bigl(\log \frac{1}{\delta } + \frac{V}{\mu \delta}\Bigr)$ \\
 \hline 
  shot and gate noise  & $\widetilde{\Oscr}\Bigl(\log \frac{1}{\delta -B^{\Escr}\mu} + \frac{V^{\Escr}}{\mu \delta}\Bigr)$\\
 \hline 
 shot and gate noise with QEM  & $\widetilde{\Oscr}\Bigl(\log \frac{1}{\delta } + \frac{V^{\QEM}}{\mu \delta}\Bigr)$ \\
 \hline
\end{tabular}
\end{minipage} \hfill
\begin{minipage}[c]{0.45 \linewidth}
     \centering
     \begin{tabular}{|c|c|}
\hline
   \textbf{Parameters}  &  \textbf{Scaling} \\
     \hline 
 variance $V$   & $\Oscr(D/N_m)$ \\
 \hline 
 bias $B^{\Escr}$  & $\Oscr(D \gamma)$ \\
 \hline 
 variance $V^{\Escr}$  &  $\Oscr(D c(\gamma)/N_m)$ \\
 \hline
 variance $V^{\QEM}$  & $\Oscr(c_1(\gamma)V^{\Escr})$ +$\Oscr(c_2(\gamma) D/N_c)$\\
 \hline
\end{tabular}
     \end{minipage} 
\captionof{table}{(Left) Upper bounds on the number of iterations $T$ required to ensure a target error floor, $\Ebb[L(\thetabf^T)]-L(\thetabf^{*})$, of order $\Oscr(\delta)$ for any $\delta>0$  when only shot noise is present, as well as in the presence of both shot and gate noise with and without QEM.
Parameters $V, V^{\Escr}$ and $V^{\QEM}$ quantify the respective variances of the gradient estimates;
 $\mu>0$ is a positive constant dependent on the ansatz and on the observable $H$; $B^{\Escr}>0$ accounts for the bias due to quantum gate noise; and functions $c(\gamma)$, $c_1(\gamma)$, and $c_2(\gamma)$ are detailed in Sec.~\ref{sec:properties_noisyestimator} and Sec~\ref{sec:QEM_properties}. 
 } \label{tab:results}
\end{figure*}
This paper aims at providing some theoretical insights on the potential advantages of QEM for VQE by studying the convergence of SGD for problem \eqref{eq:VQE} with and without QEM. Assuming that gradients are estimated via the \textit{parameter-shift rule} \cite{schuld2015introduction}, the specific contributions are summarized in Table~\ref{tab:results}, where we report the derived upper bounds on the number of SGD iterations required to ensure convergence to a target error floor of order $\Oscr(\delta)$\footnote{In this work , we use the standard big-O notation $\Oscr(\cdot)$, with notation $\widetilde{\Oscr}(\cdot)$ further hiding the poly-logarithmic factors \cite{ajalloeian2020convergence}.} for any $\delta>0$. We compare the performance of SGD with shot noise only; as well as with shot and gate noise with and without QEM. For all schemes, as per the  notation detailed in Table 1, we fix the number of model parameters $\theta$ to integer $D$; the number of measurements per-iteration as $2DN_m$ for some integer $N_m$; and we quantify the noise level with parameter $\gamma \in [0,1]$. Noise is modelled as an Pauli quantum channel that can account for spatial correlations across the qubits \cite{braccia2022quantum}. The ansatz and the observable $H$ determine a constant $\mu$ (see Sec.~\ref{sec:convergence_assumptions} for details). The main results are described as follows.
\begin{itemize}
    \item We first quantify the impact of the inherent randomness of quantum measurements, also known as \textit{shot noise}, as well as of the bias induced by \textit{quantum gate noise}, on the convergence of SGD.   We show that, while the impact of shot noise -- quantified by the variance term $V$ in Table~\ref{tab:results} -- can be arbitrarily decreased by increasing the number, $N_m$, of quantum measurements, the presence of quantum gate noise induces a \textit{non-zero error floor} on the convergence of SGD. The error floor is caused by the bias in the estimate of the gradient, as quantified by the term $B^{\Escr}$ in Table~\ref{tab:results}, which is proved to depend on the strength of the quantum gate noise $\gamma$, on the number $D$ of free parameters of the PQC, and on the eigenspectrum of the observable $H$. Accordingly, the derived upper bound on the number of SGD iterations in the presence of gate noise diverges  for any error floor $\delta$ smaller than some level, $B^{\Escr}\mu$, dependent on the bias.
    \item  To mitigate the error floor induced by quantum gate noise, we then study the impact of  QEM on the convergence of SGD. 
   As seen in Table~\ref{tab:results}, QEM can obtain any error $\delta>0$ in a finite number of iterations thanks to the elimination of the bias on the gradient estimate caused by quantum gate noise. Therefore, in order to obtain an error smaller than $B^{\Escr}\mu$, QEM is necessary. For levels of the error larger than $B^{\Escr}\mu$, QEM can help reduce the number of required SGD iterations. According to the bounds derived in this paper, this is the case if the noise level $\gamma$ is sufficiently small and if the number $N_c$ of circuits sampled at each iteration is large enough. In fact, as illustrated in Fig. \ref{fig:effectofQEM}, QEM causes an increase in the variance of the gradient estimator, as quantified by the term $V^{\QEM}$ in Table~\ref{tab:results}. In particular, QEM adds a contribution to the variance that decreases as $1/N_c$, requiring a sufficiently large value of $N_c$ to ensure that this term does not dominate the overall performance. 
    \item We corroborate the theoretical findings with numerical experiments on the VQE problem of maximizing the weighted max-cut Hamiltonian \cite{amaro2022filtering}.
\end{itemize}

\subsection{Related Works}
Recent works \cite{kandala2017hardware,kubler2020adaptive} provide numerical evidence for the detrimental impact of the quantum gate noise on the trainability of PQCs. Importantly, reference \cite{wang2021noise} shows that, in the presence of local Pauli noise,  if the circuit depth is linear in the number of qubits,  the gradient of the cost function vanishes exponentially in the number of qubits -- a phenomenon termed \textit{noise-induced barren plateaus}. This suggests that PQC optimization can indeed benefit from QEM. The recent work \cite{kandala2019error} demonstrates the enhanced performance of VQE with error mitigation via zero noise extrapolation on a superconducting quantum processor.  Resorting to error mitigation via hidden inverses \cite{zhang2022hidden}, the work \cite{leyton2022quantum} proves experimentally that error-mitigated VQE converges faster when applied to problems in quantum chemistry.

From a theoretical standpoint, the impact of gate noise on the convergence of SGD for the VQE was studied in
\cite{gentini2020noise}, where the authors consider gradient estimators based on symmetric logarithmic derivative operators. The variance of the resulting gradient estimators is quantified in terms of the quantum Fisher information, which is shown to decrease with increasing gate noise. In contrast, our focus in this work is on SGD schemes that leverage the more commonly used parameter-shift rule-based gradient estimators. 

The authors in \cite{du2021learnability} study the trainability of PQCs for empirical risk minimization in supervised learning by investigating the impact of gate noise on the convergence of SGD. However, they do not explicitly characterize the non-zero error floor induced by the quantum gate noise. Furthermore,
none of the works summarized above account for the impact of QEM. 

The references \cite{wang2021can,wang2021noise} tackle the problem of noise-induced barren plateaus by studying the question of whether error mitigation techniques can improve the resolvability of any two points on the cost landscape. Under global depolarizing noise, QEM is shown to improve  resolvability,  provided the number of qubits is small. However, the  improvement in resolvability due to QEM degrades exponentially for large problem sizes (number of qubits, circuit depth) in the presence of local depolarizing noise.
In general, the authors conclude that the exponential cost concentration induced by noise cannot be mitigated without committing additional resources.  



Overall, to the best our knowledge, no prior work has addressed the convergence performance of SGD-based VQE with the parameter shift rule and in the presence of quantum error mitigation.

\subsection{Organization}
The rest of the paper is organized as follows. In Section~\ref{sec:problem_formulation}, we detail the VQE problem  \eqref{eq:VQE}, the solution method based on SGD, as well as the considered quantum gate noise model. Section~\ref{sec:shot_noise} presents the convergence analysis of SGD for VQE when gate noise is negligible. The impact of gate noise on the convergence of SGD is studied in Section~\ref{sec:gate_noise}. In Section~\ref{sec:QEM}, we analyze the convergence of SGD with QEM. Finally, Section~\ref{sec:experiments} presents numerical experiments that corroborate our analytical findings.
\section{Problem Formulation}\label{sec:problem_formulation}
 The goal of the variational quantum eigensolver (VQE) is to prepare a quantum state, via a parameterized quantum circuit (PQC), that minimizes the expectation of an observable. This problem is also at the core of quantum machine learning algorithms, and is the focus of this paper. This section first introduces the VQE problem, as well as the standard solution method based on stochastic gradient descent (SGD). Then, we describe the noise model assumed for the quantum circuit. 
\subsection{Variational Quantum EigenSolver}
An ideal, noiseless,  PQC implements a unitary transformation $U(\thetabf)$, operating on $n$ qubits, that is parameterized by  a vector $\thetabf=(\theta_1,\hdots, \theta_{D})$, with $\theta_d\in \mathbb{R}$ for $d=1,...,D$. A common architecture for the PQC, also known as ansatz, prescribes unitaries of the form \cite{sweke2020stochastic}
\begin{align}
U(\thetabf)=\prod_{d=1}^{D} U_d(\theta_d)V_d, \label{eq:overall_unitary}
\end{align}where the $d$th parameter $\theta_d$ in vector $\thetabf$ determines the operation of the $d$th unitary $U_d(\theta_d)$. The parameterized gate is  assumed to be given as \cite{benedetti2019parameterized} 
\begin{align}
    U_d(\theta_d)= \exp \biggl( -i \frac{\theta_d}{2}G_d\biggr),
\end{align}where $G_d \in \{I,X,Y,Z\}^{\otimes n}$ denotes the Pauli string generator. Furthermore, the unitary $V_d$ in \eqref{eq:overall_unitary} is fixed, and not dependent on parameters $\theta$. 

The noiseless PQC is applied to $n$ qubits that are initially in the ground state $\vert 0 \rangle$  to yield the parameterized quantum state 
\begin{align}
    \vert \Psi(\thetabf)\rangle =U(\thetabf)\vert 0 \rangle  \label{eq:quantumstate}.
\end{align}  The pure-state density matrix corresponding to the quantum state \eqref{eq:quantumstate} is denoted as \begin{align}\Psi(\thetabf)=\vert \Psi(\thetabf)\rangle \langle \Psi(\thetabf) \vert.
\end{align}

The VQE seeks to find the parameter vector $\thetabf \in \mathbb{R}^{D}$ that addresses the problem \eqref{eq:VQE}, which we restate here as
\begin{align}
       \min_{\thetabf \in \Real^D} \biggl \lbrace L(\thetabf):=\langle H \rangle _{\vert \Psi(\thetabf) \rangle}  \biggr \rbrace \label{eq:problem1},
    \end{align} where  the loss function $L(\thetabf)$ is the expected value \begin{align}\langle H \rangle _{\vert \Psi(\thetabf) \rangle}= \langle 0 \vert  U^{\dag}(\thetabf)H U(\thetabf) \vert 0 \rangle = \Tr(H \Psi(\thetabf)) \end{align} of an observable $H$ for the state $\vert \Psi(\thetabf)\rangle$ produced by the PQC. We denote as
    \begin{align}
        L^{*}=\min_{\thetabf \in \Real^D} L(\thetabf)
    \end{align} the minimum value of the loss function.
    
    The Hermitian matrix $H$ describing the observable can be written via its eigendecomposition as
    \begin{align}
 H= \sum_{y=1}^{N_h}  h_y \Pi_y, \end{align} where  $\{h_y\}_{y=1}^{N_h}$ denote the $N_h \leq 2^n$ distinct eigenvalues of matrix $H$, and $\{\Pi_y\}_{y=1}^{N_h}$ are the projection operators onto the corresponding eigenspaces. Consequently, the loss function $L(\thetabf)$ in problem \eqref{eq:problem1} can be expressed as
\begin{align}
   L(\thetabf)=\sum_{y=1}^{N_h} h_y  \Tr(\Pi_y \Psi(\thetabf)) \label{eq:problem2}.
\end{align}


\subsection{Stochastic Gradient Descent}
Assuming the ideal case in which the PQC is noiseless,  problem \eqref{eq:problem1} can be addressed by following a hybrid quantum-classical optimization approach. Accordingly,  $N_m$ measurements of the observable $H$ are made for the output state $\vert \Psi(\thetabf) \rangle$ of PQC, producing samples $\{H_j\}_{j=1}^{N_m}$ with $H_j \in \{h_1,\hdots, h_{N_h}\}$. From these samples, an estimate of the loss function $L(\thetabf)$ is obtained as
\begin{align}
    \hat{L}(\thetabf)=\widehat{\langle H \rangle}_{\vert \Psi(\thetabf)\rangle}=\frac{1}{N_m}  \sum_{j=1}^{N_m} H_j, \label{eq:shotnoise_observable}
\end{align} where the ``hat" notation is used for empirical estimates. The empirical estimate \eqref{eq:shotnoise_observable} is used as input to classical optimization algorithm. 

In this work, we focus on the standard implementation of VQE based on SGD optimization. Starting from an initialization $\thetabf^{0} \in \mathbb{R}^D$, SGD performs the  iterative update of the parameter $\thetabf$ as
\begin{align}
    \thetabf^{t+1}&=\thetabf^t -\eta_t \hat{g}_t ,\label{eq:SGD}
\end{align} for iterations $t=1,\hdots, T$, where $\hat{g}_t$ represents a stochastic estimate of the gradient
\begin{equation}\nabla L(\thetabf)|_{\thetabf=\thetabf^t}=\begin{bmatrix}
\frac{\partial L(\thetabf)}{\partial \theta_1} \\
\vdots\\
\frac{\partial L(\thetabf)}{\partial \theta_D} \label{eq:gradient_vector}
\end{bmatrix}_{\thetabf=\thetabf^t}\end{equation} of the loss function $L(\thetabf)$ at the current iterate $\thetabf=\thetabf^t$, and $\eta_t>0$ denotes the learning rate at the $t$th iteration. We will detail how to obtain the estimate $\hat{g}_t$ as a function of measurements of observable $H$ of the form \eqref{eq:shotnoise_observable} in the next sections. 
\subsection{Noisy Quantum Gates}\label{sec:noisy_gates}
In the current era of NISQ hardware, the implementation of the unitary gates  that determine the operation of the PQC in \eqref{eq:overall_unitary} is subject to quantum gate noise due to decoherence as well as gate infidelities.
Therefore, a quantum computer can in practice only implement a given set $\Omega$ of, generally
noisy, quantum operations. 

The set $\Omega=\{\Oscr_i^{\theta} \}$ consists of a set of gates $\{\Oscr_i^{\theta}\}$ that can be implemented  on the quantum device. The notation $\Oscr_i^{\theta}$ emphasizes the possible dependence of operation $\Oscr$ on a model parameter $\theta$. Each gate $\Oscr_i^{\theta}$ can be modelled as a \textit{quantum channel}, that is, as a completely positive trace preserving (CPTP) map between density matrices. Accordingly, we can write the corresponding maps as $\Oscr_i^{\theta}(\rho)=\rho'$, where $\rho$ is the input density matrix and $\rho'$ the output density matrix \cite{nielsen2002quantum}.  The operations in set $\Omega$ can be determined via
quantum gate set tomography as detailed in \cite{endo2018practical}.

Each  unitary gate $U_d(\theta_d)V_d$ in the ideal, noiseless, PQC \eqref{eq:overall_unitary} corresponds to the noiseless quantum channel \begin{align}\rho'=\Uscr_d^{\theta_d}(\rho)=V_dU_d(\theta_d)\rho U_d(\theta_d)^{\dag}V_d^{\dag}.\label{eq:gate_operation}\end{align} Thus, the operation of the overall noiseless PQC $U(\thetabf)$ in \eqref{eq:overall_unitary} on the initial state $\rho_0=\vert 0 \rangle \langle 0 \vert $ can be expressed as the composition of the mappings
\begin{align}
    \Uscr^{\thetabf}(\rho_0)=\Uscr^{\theta_{D}}_{D} \circ \hdots \circ \Uscr^{\theta_1}_1(\rho_0),\label{eq:overallunitaryoperation}
\end{align}where $\circ$ indicates the composition operation.


In practice, the quantum gates in the ideal PQC \eqref{eq:overallunitaryoperation}
can be only approximately implemented on the quantum computer, in the sense
that the set of feasible operations $\Omega$ includes only noisy versions of
such gates. Formally, given the noiseless parameterized gate $\Uscr_d^{\theta_d}(\cdot)$ in \eqref{eq:gate_operation}, we assume that the set $\Omega$ includes a noisy operation of the form \begin{align}\Oscr_i^{\theta_d}(\cdot) = \Escr \circ \Uscr_d^{\theta_d}(\cdot) = \widetilde{\Uscr}_d^{\theta_d}(\cdot), \label{eq:noisygate} \end{align} where $\Escr(\cdot)$ represents a quantum channel. Following \eqref{eq:noisygate}, we will write as $\widetilde{\Uscr}_d^{\theta_d}(\cdot)$ the noisy quantum gate corresponding to the noiseless gate $\Uscr_d^{\theta_d}(\cdot)$ in \eqref{eq:gate_operation}.

Throughout this work, we model gate noise via \textit{Pauli quantum channels} acting on the $n$ qubits of the form
\begin{align}
    \Escr(\rho)=(1-\epsilon)\rho+\epsilon\sum_{j}E_j \rho E_j^{\dag}, \label{eq:noisemodels}
\end{align}  where $0\leq \epsilon \leq 1$ is the error probability, and the sum in \eqref{eq:noisemodels} runs over a subset of $4^n-1$ $n$-length Pauli strings $E_j$, excluding the identity matrix. The noise model in \eqref{eq:noisemodels} includes the standard model with independent channels across the $n$ qubits, and it can also more generally account for spatial correlations across $n$ qubits \cite{braccia2022quantum}. It is practically well justified when quantum gate set tomography returns a real, diagonal matrix for the Pauli transfer matrix representation \cite{mangini2021qubit}.

We write  $P_s$, with $s \in \{0,1,2,3\}$, for a single-qubit Pauli operator, where $P_0=I$, $P_1=X$, $P_2=Y$ and $P_3=Z$. With this notation,  the Kraus operators in (\ref{eq:noisemodels}) can be written as \begin{equation}E_j=\sqrt{p_j} P_{s_{j,0}} \otimes \hdots \otimes P_{s_{j,n-1}}, \end{equation} where $s_{j,k} \in \{0,1,2,3\}$ for $k=0,1,\hdots,{n-1}$; we have the equality $\sum_j p_j =1$; and the indices $s_{j,k}$ cannot be all equal to $0$.

As an important example that we will assume in some of the theoretical derivations,  the noise model \eqref{eq:noisemodels} includes the  \textit{depolarizing noise channel}, which is obtained when the Kraus operators $E_j$ include all $4^n -1$ Pauli strings excluding the identity matrix, and $p_j=1/(4^n -1)$. Note that this corresponds to a situation in which the gate noise applies separately to each qubit.

Overall, when QEM is not applied, the actual implementation of PQC \eqref{eq:overall_unitary} on the quantum
computer amounts to the following composition of noisy gates
\begin{align}
    \widetilde{\Uscr}^{\thetabf}(\cdot) = \widetilde{\Uscr}_D^{\theta_D} \circ \hdots \circ \widetilde{\Uscr}_1^{\theta_1}(\cdot), \label{eq:noisyPQC}
\end{align}
where each operation $\widetilde{\Uscr}_d^{\theta_d}(\cdot)$, for $d=1,\hdots,D$, belongs to the family of quantum operations $ \Omega$. Accordingly, the PQC operation of \eqref{eq:noisyPQC} on the initial quantum state $\rho_0=\vert 0 \rangle \langle 0 \vert$ results in the output noisy quantum state 
\begin{align}
    \rho^{\Escr}(\thetabf)=\widetilde{\Uscr}^{\thetabf}(\rho_0). \label{eq:noisyquantumstate}
\end{align}

\subsection{Feasible Quantum Gate Set}\label{sec:feasiblegate}

Apart from the noisy gates \eqref{eq:noisygate}, in order to enable QEM, we assume that the set $\Omega$ of implementable operations also includes the noisy cascade of Pauli operations and ideal gates ${\Uscr}^{\thetabf}_d(\cdot)$. To denote an $n$-qubit Pauli operation, we henceforth use the following standard convention.
Let $\sbar=(s_0,\hdots,s_{n-1})$ be a vector of indices with $s_k \in \{0,1,2,3\}$ for $k=0,\hdots,n-1$.  Using this notation, we denote an $n$-qubit Pauli operator as \begin{align}\Pscr_{\sbar}(\cdot)= \Pscr_{s_0}(\cdot) \otimes \hdots \otimes \Pscr_{s_{n-1}}(\cdot)\label{eq:Paulistring}\end{align} where \begin{equation}\Pscr_{s_k}(\cdot)=P_{s_k}(\cdot)P_{s_k}^{\dag}\end{equation} for $k=0,\hdots,n-1$.
The set of implementable operations  $\Omega$ then includes all the noisy gates $\widetilde{\Uscr}^{\thetabf}_d(\cdot)$, as well as the noisy concatenations $\Escr \circ \Pscr_s \circ \Uscr^{\thetabf}_d$ of Pauli operation $\Pscr_s$ on the ideal unitary for all $s \in\{0,1,2,3\}^{\otimes n}$.

\section{Impact of Shot Noise}\label{sec:shot_noise}
We start our analysis of the SGD-based optimization for the VQE problem \eqref{eq:problem1} by assuming that the quantum gate noise is negligible. This amounts to assuming that the set $\Omega$ of feasible quantum operations includes all gates $\Uscr_d^{\theta_d}(\cdot)$ in \eqref{eq:overallunitaryoperation}, and hence we have $\widetilde{\Uscr}_d^{\theta_d}(\cdot)= \Uscr_d^{\theta_d}(\cdot)$ in  \eqref{eq:noisygate}. Accordingly, the analysis in this section only accounts for the randomness due to shot noise, that is, due to the inherent stochasticity of quantum measurements. The impact of gate noise  will be studied in the next section. 
\subsection{Estimating the gradient}
With no gate noise, the $d$th component  of the gradient, $[g_t]_d=[\nabla L(\thetabf^t)]_d$, can be estimated 
via the \textit{parameter shift-rule} as \cite{schuld2019evaluating}
\begin{align}
   [\hat{g}_t]_d= \frac{1}{2} \Bigl(  \widehat{\langle H \rangle}_{\vert \Psi(\thetabf^t+\frac{\pi}{2}{e}_d) \rangle}- \widehat{\langle H \rangle}_{\vert \Psi(\thetabf^t-\frac{\pi}{2}{e}_d)\rangle}\Bigr), \label{eq:unbiasedgradient_1}
\end{align} where ${e}_d$ denotes a unit vector with all zero elements except in the $d$th position. In \eqref{eq:unbiasedgradient_1}, the quantum state $\vert \Psi(\thetabf^t +\frac{\pi}{2}{e}_d)\rangle$ is obtained by shifting the phase of $d$th parameter $\theta_d$ by $ \pi/2$; and $\vert \Psi(\thetabf^t -\frac{\pi}{2}{e}_d)\rangle$ is obtained by shifting the phase of $d$th parameter $\theta_d$ by $ -\pi/2$. Furthermore, as in \eqref{eq:shotnoise_observable}, the notation
 $\widehat{\langle H \rangle}_{\vert \Psi(\thetabf) \rangle}$ describes an empirical estimation of the expected observable $H$, obtained from $N_m$ i.i.d. measurements $\{H_1,\hdots,H_{N_m}\}$ of the observable under state $\vert \Psi(\thetabf) \rangle$.
 
 The measurements of observable $H$ return $H_j=h_y$ with probability \begin{align}
    p(y|\thetabf)=\Tr(\Pi_y \Psi(\thetabf)), \quad y \in \{1,\hdots,N_h\} \label{eq:probability_shotnoise}
\end{align} for all measurements $j=1,\hdots, N_m$.
The resulting empirical estimator of each expected value in \eqref{eq:unbiasedgradient_1} evaluates as in \eqref{eq:shotnoise_observable}. Therefore, estimating the full gradient vector \eqref{eq:gradient_vector}
 requires running the PQC a number of times equal to $2DN_m$.
 
Upon obtaining a solution $\thetabf^T$ after $T$ iterations, the final outcome of the VQE is evaluated as $\langle H \rangle_{\vert \Psi(\thetabf^T) \rangle}$, where the expectation is taken over a large number of measurements.
 
 \subsection{Properties of the Gradient Estimte}

 
The estimate in \eqref{eq:unbiasedgradient_1} is unbiased, \ie, we have \begin{align}\Ebb[\hat{g}_t]=g_t=\nabla L(\thetabf^t), \label{eq:g_t}
\end{align}
where the expectation $\Ebb[\cdot]$ is taken with respect to the $2DN_m$ measurements used in \eqref{eq:unbiasedgradient_1}.  Therefore, we can decompose the gradient estimator \eqref{eq:unbiasedgradient_1} as
\begin{align}
    \hat{g}_t= g_t + \xi_t, \label{eq:gradient_estimate}
\end{align} with noise $\xi_t$ satisfying the conditions $\Ebb[\xi_t]=0$ and $\mathrm{var}(\xi_t)=\Ebb[\lVert\hat{g}_t-g_t\rVert^2]$. 
The following lemma gives an upper bound on the variance $\mathrm{var}(\xi_t)$. To this end, we first rewrite the empirical estimate \eqref{eq:shotnoise_observable} as \begin{align}
    \widehat{\langle H \rangle}_{\vert \Psi(\thetabf)\rangle} = \frac{1}{N_m}  \sum_{j=1}^{N_m} \sum_{y=1}^{N_h}h_y\Ibb\{Y_j=y\}, \label{eq:shot_noise_observable_new}
\end{align} where $Y_j \in \{1,\hdots,N_h\}$ is the random variable defining the index corresponding to the $j$th measurement so that we have the equality $H_j=h_{Y_j}$. Furthermore, we introduce the variance of the Bernoulli quantum measurement  $\Ibb\{Y_j=y\}$ as \begin{align}
    \nu(p(y|\thetabf))=p(y|\thetabf)(1-p(y|\thetabf)), \label{eq:variance}
\end{align} where $p(y|\thetabf)$ is as defined in \eqref{eq:probability_shotnoise}.
\begin{lemma} \label{lem:1}
The variance of the unbiased gradient estimator $\hat{g}_t$ in \eqref{eq:unbiasedgradient_1} can be upper bounded as
\begin{align}
    \mathrm{var}(\xi_t)  \leq  \frac{\nu N_h D \Tr(H^2)}{2N_m}:=V \label{eq:variance_noiseless},
\end{align} where $\nu \in [0,0.25]$ is  \begin{align}\nu= \max_{\thetabf \in \Real^D} \max_{y \in \{1,\hdots,N_h\}} \nu(p(y|\thetabf)) \label{eq:nu2}, \end{align} with $\nu(p(y|\thetabf))$ in \eqref{eq:variance}.
\end{lemma}
\begin{proof}
The proof is in Appendix~\ref{app:1}.
\end{proof}

Lemma~\ref{lem:1} shows that the variance of the gradient estimate \eqref{eq:gradient_estimate} decreases (at least) as $1/(2N_m)$ with respect to the number of per-parameter measurements $2N_m$, while being proportional to the number of parameters, $D$. The term $\nu$ in \eqref{eq:nu2} captures the maximum variance of the Bernoulli measurement indicator variables $\Ibb\{Y_j=y\}$, hence accounting for the degree of randomness of the measurements. Finally, the term $\Tr(H^2)$ is a scale parameter dependent on the numerical values $\{h_y\}_{y=1}^{N_h}$ of the observable $H$.

\subsection{Convergence of SGD}\label{sec:convergence_assumptions}
 We now study the convergence properties of  SGD in \eqref{eq:SGD} with the unbiased estimator \eqref{eq:gradient_estimate} in the case under study of noiseless gates.
Towards this goal, we make the following assumptions on the loss function as in \cite{sweke2020stochastic, gentini2020noise}. These assumptions are further discussed in Appendix~\ref{app:smoothness}.
\begin{assumption}\label{assum:1}
The loss  function $L(\thetabf)$ is $\Lscr$-smooth, \ie, we have the inequality
\begin{align}
    L(\thetabf) \leq L(\thetabf') + \nabla L(\thetabf')^T(\thetabf-\thetabf') +\frac{\Lscr}{2} \lVert\thetabf-\thetabf'\rVert ^2
\label{eq:assum_smoothness}\end{align} for all $\theta, \theta' \in \Theta$. Furthermore, it satisfies the  $\mu$-Polyak-Lojasiewicz (PL) condition, \ie, there exists a constant $ \mu>0$ and $\mu \leq \Lscr$ such that the inequality
\begin{align}
  \lVert \nabla L(\thetabf) \rVert^2 \geq 2 \mu (L(\thetabf)-L^{*})
\end{align}holds for all $\thetabf \in \Theta$.
\end{assumption}In practice, the constant $\mu$ depends on the number $n$ of qubits and the number of gate parameters $D$, since a larger circuit is typically characterized by smaller-gradient norms \cite{mcclean2018barren}.

The following result gives a bound on the optimality of the SGD output $\thetabf^T$ after $T$ iteration. The proof of the result is based on the classical convergence analysis of SGD, and can be found  in \cite[Theorem 6]{ajalloeian2020convergence}. 
\begin{theorem}
Under Assumption~\ref{assum:1}, for any given initial point $\thetabf^{0}$, the following bound holds for any fixed learning rate $\eta_t=\eta \leq 1/\Lscr$:
\begin{align}
    \Ebb[L(\thetabf^T)]-L^{*} \leq (1-\eta \mu)^T (\Ebb[L(\thetabf^0)]-L^{*})+\frac{1}{2} \biggl[\frac{\eta \Lscr V}{\mu} \biggr], \label{eq:convergence_1}
\end{align}
where $V$ is as defined in \eqref{eq:variance_noiseless}, and the expectation is taken over the distribution of the measurement outputs. Furthermore, given some target error level $\delta >0$, for learning rate
$\eta =\eta^{\mathrm{shot-noise}}\leq  \min \{\frac{1}{\Lscr}, \frac{\delta \mu}{\Lscr V}\}$, a number of iterations \begin{align}T^{\mathrm{shot-noise}}=\tilde{\Oscr}\biggl(\log \frac{1}{\delta} + \frac{V}{\delta \mu}\biggr)\frac{\Lscr}{\mu}\end{align} is sufficient to ensure an error $\Ebb[L(\thetabf^T)]-L^{*} =\Oscr(\delta)$.
\end{theorem}

The result in \eqref{eq:convergence_1} shows that SGD with the unbiased gradient estimator \eqref{eq:gradient_estimate} converges up to an error floor of the order $\Oscr(\eta \Lscr V)$. By choosing the learning rate $\eta$ sufficiently small, this error can be made arbitrarily small, thereby guaranteeing  convergence in at most $T^{\mathrm{shot-noise}}$ iterations.

\section{Impact of Gate Noise}\label{sec:gate_noise}
In this section, we study the impact of gate noise. To this end, we assume that all quantum gate $\widetilde{\Uscr}_d^{\theta_d}(\cdot) \in \Omega$ in the PQC \eqref{eq:overallunitaryoperation} can be written as in \eqref{eq:noisygate} with quantum noise \eqref{eq:noisemodels}.  As we will see, 
 the presence of noisy gates induces a bias in the gradient estimator   \eqref{eq:gradient_estimate} studied in the previous section. The main result of this section quantifies the impact of the bias in the gradient estimate on the convergence of SGD in \eqref{eq:SGD}. 
 \subsection{Estimating the Gradient}
 To start, in a manner similar to \eqref{eq:shotnoise_observable}, let us denote as $\widehat{\langle H \rangle}_{\rho^{\Escr}(\thetabf)}$, the empirical estimator of the observable $H$ obtained from measurements of the state $\rho^{\Escr}(\thetabf)$ in \eqref{eq:noisyquantumstate} produced by the noisy PQC. This estimate is obtained from $N_m$ i.i.d. measurements $\{H^{\Escr}_1,\hdots, H^{\Escr}_{N_m}\}$ of observable $H$ under the noisy quantum state $\rho^{\Escr}(\thetabf)$. We also write as $\{Y_1^{\Escr},\hdots, Y_{N_m}^{\Escr}\}$, with $Y_j^{\Escr} \in \{1,\hdots,N_h\}$, the corresponding indices of the measurement outcomes, so that we have $H^{\Escr}_j=h_{Y_j^{\Escr}}$. Accordingly, the output $Y^{\Escr}_j=y$ is produced with the probability 
\begin{align}
    p^{\Escr}(y|\thetabf)=\Tr(\Pi_y \rho^{\Escr}(\thetabf)). \label{eq:probability_gatenoise}
\end{align}

At $t$th iteration, following the parameter shift rule as in \eqref{eq:unbiasedgradient_1}, we adopt the stochastic gradient estimator $\hat{g}^{\Escr}_t$  whose $d$th component is obtained as
\begin{align}
    [\hat{g}^{\Escr}_t]_d= \frac{1}{2} \biggl( \widehat{\langle H \rangle}_{\rho^{\Escr}(\thetabf^t+\frac{\pi}{2} {e}_d)} - \widehat{\langle H \rangle}_{\rho^{\Escr}(\thetabf^t-\frac{\pi}{2} {e}_d)}\biggr) \label{eq:biased_estimator}.
\end{align}  The estimate \eqref{eq:biased_estimator} involves the noisy quantum states $\rho^{\Escr}(\thetabf^t\pm\frac{\pi}{2}{e}_d)$ in \eqref{eq:noisyquantumstate} that are produced by the noisy PQC with $d$th parameter $\theta_d$  phase-shifted by $ \pi/2$ or $-\pi/2$. In a manner similar to equation \eqref{eq:variance}, we  define  the variance of the Bernoulli quantum measurement $\Ibb\{Y^{\Escr}_j=y\}$ as \begin{align}\nu(p^{\Escr}(y|\thetabf))=p^{\Escr}(y|\thetabf)(1-p^{\Escr}(y|\thetabf)).\end{align}

\subsection{Properties of the Gradient Estimate}\label{sec:properties_noisyestimator}
The stochastic gradient estimator 
in \eqref{eq:biased_estimator} is affected by two sources of randomness, namely the shot noise due to quantum measurements as studied in the previous section, and the distortion caused by the noisy quantum gates. While shot noise does not cause a bias in the gradient estimate \eqref{eq:biased_estimator}, the quantum gate noise causes the estimate \eqref{eq:biased_estimator} to be biased, \ie, $\Ebb[\hat{g}^{\Escr}_t]\neq g_t$ where  $g_t=\nabla L(\thetabf^t)$ is the exact gradient at iteration $t$ as in \eqref{eq:g_t}.

To elaborate on this point, we decompose the biased gradient estimator in \eqref{eq:biased_estimator} as
\begin{align}
   \hat{g}^{\Escr}_t=g_t + \underbrace{\biggl(g_t^{\Escr}-g_t \biggr)}_{\mathrm{bias}} +\xi^{\Escr}_t, \label{eq:noisy_decomposition}
\end{align} where the noise term is zero mean, \ie,  $\Ebb[\xi^{\Escr}_t]=0$, and we have defined \begin{align}
[g^{\Escr}_t]_d=\frac{1}{2} \biggl( \langle H \rangle_{\rho^{\Escr}(\thetabf^t+\frac{\pi}{2}{e}_d)} - \langle H \rangle_{\rho^{\Escr}(\thetabf^t-\frac{\pi}{2}{e}_d)}\biggr).
\label{eq:unbiased_noisygradient}
\end{align} The gradient \eqref{eq:unbiased_noisygradient} represents the average of the estimate \eqref{eq:biased_estimator}, where the expectation is taken over the quantum measurements. Consequently, the difference $(g_{t}^{\Escr}-g_t)$ in \eqref{eq:unbiased_noisygradient} captures the bias due to gate noise.

The following lemma presents bounds on the bias and on the variance of the measurement noise of the gradient estimator in \eqref{eq:biased_estimator}. To proceed, we write the noisy quantum state in \eqref{eq:noisyquantumstate} in terms of the following decomposition \cite{koczor2021dominant}
\begin{align}
    \rho^{\Escr}(\thetabf)=(1-\gamma) \Psi(\thetabf)+\gamma \tilde{\rho}(\thetabf), \label{eq:decomposition_density}
\end{align}where parameter $\gamma=1-(1-\epsilon)^{D}$ describes the level of noise. In particular, we have $\gamma=0$ if there is no gate noise $(\epsilon=0)$, and $\gamma=1$ when the noise level is maximal. The decomposition \eqref{eq:decomposition_density} amounts to a convex combination of the ideal state density matrix $\Psi(\thetabf)=\vert \Psi(\thetabf) \rangle \langle \Psi(\thetabf) \vert$ and of the \textit{error density matrix} $\tilde{\rho}(\thetabf)$  defined as \begin{align}
\tilde{\rho}(\thetabf)=\frac{1}{\gamma}(\rho^{\Escr}(\thetabf)-(1-\gamma) \Psi(\thetabf)).\label{eq:errordensity}
\end{align} Note that, under the noise model of \eqref{eq:noisemodels}, matrix $\tilde{\rho}(\thetabf)$ is indeed a valid density operator for $\gamma>0$ \cite{koczor2021dominant}.

Let $\tilde{Y} \in \{1,\hdots, N_h\}$ denote the index of the quantum measurement $\tilde{H}=h_{\tilde{Y}}$ of the observable $H$ under the error density matrix $\tilde{\rho}(\theta)$ in \eqref{eq:decomposition_density}. The probability of observing $\tilde{Y}=y \in \{1, \hdots, N_h\}$ is given as \begin{align}
    \tilde{p}(y|\theta)=\Tr(\Pi_y \tilde{\rho}(\theta)) \label{eq:probability_errordensity}.
\end{align} In a manner similar to \eqref{eq:variance}, we define the variance of the Bernoulli quantum measurement $\Ibb\{\tilde{Y}=y\}$ as
  $$\nu(\tilde{p}(y|\thetabf))=\tilde{p}(y|\thetabf)(1-\tilde{p}(y|\thetabf)).$$  
\begin{lemma}\label{lem:biased_estimator}
  The following upper bound holds on the variance of the stochastic gradient estimator $\hat{g}^{\Escr}_t$
\begin{align}
 \mathrm{var}(\xi^{\Escr}_t)\leq \frac{D N_h \Tr(H^2)}{2N_m}  c(\gamma) := V^{\Escr},  \label{eq:variance_noisy}
\end{align}where $c(\gamma)=\max_{y\in \{1,\hdots,N_h\}} \max_{\theta \in \Real^D} \nu(p^{\Escr}(y|\thetabf))$, with
\begin{align}
    \nu(p^{\Escr}(y|\thetabf))&=\gamma \nu(\tilde{p}(y|\thetabf))+\nu((1-\gamma)) (p(y|\thetabf)-\tilde{p}(y|\thetabf))^2\non \\&+ (1-\gamma) \nu(p(y|\thetabf)), \label{eq:relation_1}
\end{align}
and $\gamma=1-(1-\epsilon)^{D}$;
 $p(y|\thetabf)$ as in \eqref{eq:probability_shotnoise}; and $\tilde{p}(y|\thetabf)$ in \eqref{eq:probability_errordensity}. Furthermore, the norm of the bias in \eqref{eq:noisy_decomposition} can be bounded as
\begin{align}
    \lVert \mathrm{bias} \rVert ^2 \leq 4 D \lVert H \rVert_{\infty}^2 \gamma:=B^{\Escr}. \label{eq:bias_noisy}
\end{align}
\end{lemma}
\begin{proof}
Proof is included in Appendix~\ref{app:biased_estimator}.
\end{proof}

The bound \eqref{eq:bias_noisy} on the bias of the estimator $\hat{g}_t^{\Escr}$  can be seen to be increasing with the noise level $\gamma$; to be proportional to the number of parameters $D$; and to depend on the spectrum of the problem Hamiltonian. In contrast,
the upper bound $V^{\Escr}$ in \eqref{eq:variance_noisy} quantifies the impact of shot noise and gate noise on the variance of the gradient estimator in \eqref{eq:biased_estimator}. By \eqref{eq:variance_noisy}, the variance of the gradient estimator decreases as $1/2N_m$ with respect to the total number of per-parameter measurements. Furthermore, unlike the bias, the variance $V^{\Escr}$ is not necessarily increasing with the noise level $\gamma$.

To see this, note that for a fixed number $N_m$ of measurements, the impact of the gate noise on the variance is captured by the term $\nu(p^{\Escr}(y|\thetabf))$ in \eqref{eq:relation_1}. In the absence of gate noise, \ie, when $\epsilon =\gamma= 0$, we have the equality $\nu(p^{\Escr}(y|\thetabf)) = \nu(p(y|\thetabf))$, whereby the variance  $V^{\Escr}$ reduces to the variance $V$ in \eqref{eq:variance_noiseless} caused solely by shot noise. The variance $\nu(p^{\Escr}(y|\thetabf))$ is not necessarily monotonically increasing with the noise $\gamma$. In fact, as shown in Appendix~\ref{app:properties},
there exists  a $\gamma^{*}(y,\thetabf) \in [0,1]$ such that the variance $\nu(p^{\Escr}(y|\thetabf))$ is a concave function of $\gamma$, increasing in the range $\gamma \in [0,\gamma^{*}(y,\thetabf)]$ and decreasing in the range $\gamma \in (\gamma^{*}(y,\thetabf),1]$.
\subsection{Convergence of SGD}
Using the results in the previous section, we now study the convergence properties of the SGD with the biased estimator $\hat{g}^{\Escr}_t$.  As in the previous section, we evaluate the performance of the obtained solution $\thetabf^T$ after $T$ iterations of the SGD update \eqref{eq:gradient_estimate} by assuming the availability of a noiseless PQC for testing. In other words, we evaluate the final performance in terms of the loss $L(\thetabf^T)$ in \eqref{eq:problem1}. In practice, this requires the application of QEM for final testing (but not during the optimization phase).
The following theorem illustrates the convergence properties of the biased gradient estimator.
\begin{theorem} \label{thm:noisy_convergence}
Under Assumption~\ref{assum:1}, for any given initial point $\thetabf^{0}$, the following bound on the optimality gap holds for the SGD with the biased gradient estimator in \eqref{eq:biased_estimator}, given any fixed learning rate $\eta_t=\eta \leq 1/\Lscr$:
\begin{align}
    \Ebb[L(\thetabf^T)]-L^{*} &\leq (1-\eta \mu)^T (\Ebb[L(\thetabf^0)]-L^{*})\non \\&+\frac{1}{2} \biggl[\frac{B^{\Escr}+\eta \Lscr V^{\Escr}}{\mu} \biggr],\label{eq:convergence_biased}
\end{align} where the expectations are taken over quantum measurements, and  $V^{\Escr}$ and $B^{\Escr}$ are defined as in \eqref{eq:variance_noisy} and \eqref{eq:bias_noisy}, respectively. Furthermore, given some target error level $\delta> B^{\Escr}/\mu>0$,  for learning rate $\eta =\eta^{\mathrm{gate-noise}} \leq  \min\{\frac{1}{\Lscr}, \frac{\delta \mu}{\Lscr V^{\Escr}}\}$, a number of iterations
\begin{align}T^{\mathrm{gate- noise}}=\tilde{\Oscr}\biggl(\log \frac{1}{\delta-\frac{B^{\Escr}}{\mu}} + \frac{V^{\Escr}}{\delta \mu }\biggr) \frac{\Lscr}{\mu} \label{eq:iterations_noisy}
\end{align}is sufficient to ensure the error floor
$\Ebb[L(\thetabf^T)]-L^{*}=\Oscr(\delta)$.
\end{theorem}

As in the bound \eqref{eq:convergence_1}, which holds for $\epsilon=0$,  the contribution due to the variance term $V^{\Escr}$ in \eqref{eq:convergence_biased} can be made arbitrarily small by keeping the learning rate $\eta$ sufficiently small. In contrast, the bias due to gate noise, which is quantified by the term $B^{\Escr}$ in \eqref{eq:bias_noisy}, prevents bound \eqref{eq:convergence_biased} from vanishing. Therefore, according to the bound \eqref{eq:convergence_biased}, the noise-induced bias causes a floor equal to $B^{\Escr}/\mu$ on the achievable error $\delta$. For error levels $\delta > B^{\Escr}/\mu$, by \eqref{eq:iterations_noisy}, 
 the induced bias $B^{\Escr}$ entails a larger number of SGD iterations to converge.

\section{Impact of Quantum Error Mitigation}\label{sec:QEM}
As seen in the previous section, the quantum gate noise  induces a bias  in the gradient estimator \eqref{eq:biased_estimator}. In this section, we  first introduce the quasi-probabilistic error mitigation (QEM) technique introduced in \cite{temme2017error}. We then study the convergence of SGD \eqref{eq:SGD} when QEM is employed in evaluating the stochastic gradient estimator on the noisy PQC.
\subsection{Quasi-Probabilistic Error Mitigation}

QEM aims to recover the ideal expected value $\langle H \rangle_{\vert \Psi(\thetabf)\rangle}=\Tr(H\Uscr^{\thetabf}(\rho_0))$ of observable $H$ under the quantum state $\Uscr^{\thetabf}(\rho_0)$  produced by the noiseless PQC. As illustrated in Fig.~\ref{fig:QEM}, this is done via a quasi-probabilistic combination of the implementable operations in the set $\Omega$ (see Section~\ref{sec:noisy_gates}).
To explain QEM, 
consider the $d$th ideal unitary map $\Uscr^{\theta_d}_d(\cdot)$. 
QEM expresses the ideal unitary operation $\Uscr_d^{\theta_d}(\cdot)$ as a  linear combination of  implementable operations in set $\Omega$  as
\begin{align}
    \Uscr_d^{\theta_d}(\cdot) =\sum_{s} q_{d,s} {\Oscr}^{\theta_d}_{s}(\cdot), \quad \textrm{with } {\Oscr}^{\theta_d}_{s} \in \Omega \label{eq:decomposition},
\end{align} where the \textit{quasi-probabilities} $q_{d,s}$ are real numbers.
Specifically, with the noise model \eqref{eq:noisygate}-\eqref{eq:noisemodels} described in Section~\ref{sec:noisy_gates}, decomposition \eqref{eq:decomposition} applies with the operations $\Oscr^{\theta_d}_{\sbar}(\cdot)= \Escr \circ \Pscr_{\sbar} \circ {\Uscr}^{\theta_d}_d (\cdot)$, for $\sbar=0,\hdots,4^n-1$,  where  $\Pscr_{\sbar}$  is the string of Pauli operations as defined in \eqref{eq:Paulistring} \cite{temme2017error}. Note that we have identified for convenience the Pauli string $\Pscr_s$ with an integer $s \in \{0,1,\hdots,4^n-1\}$, rather than with a binary string $s \in \{0,1,2,3\}^n$ as in Section~\ref{sec:feasiblegate}. This mapping between bit strings and integers is standard.

With the noise model \eqref{eq:noisygate}-\eqref{eq:noisemodels}, equation \eqref{eq:decomposition} can be then expressed equivalently via the following \textit{quasi-probabilistic representation} (QPR) \cite{temme2017error} 
\begin{align}
    \Uscr_d^{\theta_d}(\cdot) = Z_d \sum_{\sbar=0}^{4^n-1} \sgn(q_{d,\sbar}) p_{d,\sbar} {\Oscr}_{\sbar}^{\theta_d}(\cdot), \label{eq:singlegate}
\end{align}where $\Oscr^{\theta_d}_{\sbar}(\cdot)= \Pscr_{\sbar} \circ \widetilde{\Uscr}^{\theta_d}_d (\cdot)$;  \begin{align}
Z_d=\sum_{\sbar =0}^{4^n-1}|q_{d,\sbar}|  \label{eq:Z}   
\end{align} is the $l_1$-norm of the vector ${q}_d$ collecting all  quasi-probabilities for the $d$th gate; and  \begin{align}p_{d,\sbar}=\frac{|q_{d,\sbar}|}{Z_d}, \quad \mbox{for} \hspace{0.1cm} \sbar=0,\hdots,4^n-1 \label{eq:pmf}\end{align} is a probability mass function satisfying $\sum_{\sbar=0}^{4^n-1} p_{d,\sbar}=1$. By the trace-preserving property of the ideal and noisy operations, the quasi-probabilities satisfy the equality $\sum_{\sbar=0}^{4^n-1} q_{d,\sbar}=1$ and the inequality $Z_d \geq 1$.

Using \eqref{eq:singlegate}, QEM can  ideally recover the overall unitary \eqref{eq:overallunitaryoperation} implemented by the noiseless PQC as
\begin{align}
    \Uscr^{\thetabf}(\cdot)= Z \sum_{\sbfbar} p_{\sbfbar}\hspace{0.1cm}  \sgn(q_{\sbfbar}) {\Oscr}_{\sbfbar}^{\thetabf}(\cdot),\label{eq:QPR_overallunitary}
\end{align} where $\sbfbar=(\sbar_1,\hdots,\sbar_{D})$ is a $D$-dimensional vector with $\sbar_j \in \{0,\hdots,4^n-1\}$ for $j=1,\hdots,D$;
$
    {\Oscr}_{\sbfbar}^{\thetabf}(\cdot)= {\Oscr}^{\theta_D}_{\sbar_D} \circ \hdots \circ {\Oscr}^{\theta_1}_{\sbar_1}(\cdot)
$ is the corresponding noisy circuit composed of feasible operations from set $\Omega$; $Z=\prod_{d=1}^D Z_d$ is the product of the normalizing constants;  $p_{\sbfbar}=\prod_{d=1}^{D}p_{d,\sbar_d}$ is the probability of choosing the $\sbfbar$th circuit implementing operation $\Oscr^{\thetabf}_{\sbfbar}(\cdot)$; and we have $\sgn(q_{\sbfbar})=\prod_{d=1}^{D} \sgn(q_{d,\sbar_d})$ with $\sgn$ being the sign function. Note that the sum in \eqref{eq:QPR_overallunitary} is over all $4^{nD}$ values of string $\sbfbar$.
As illustrated in Fig.~\ref{fig:QEM}, QEM can thus exactly recover the true expected observable as
\begin{align}
    \Tr(H \Uscr^{\thetabf}(\rho_0))&=Z \sum_{\sbfbar}p_{\sbfbar}\sgn(q_{\sbfbar}) \langle H \rangle_{ \rho_{\sbfbar}(\thetabf)}\label{eq:H_qem},
\end{align}where \begin{align}\rho_{\sbfbar}(\thetabf)={\Oscr}^{\thetabf}_{\sbfbar}(\rho_0)\end{align} is the quantum state obtained by applying the $\sbfbar$th noisy circuit $\Oscr^{\thetabf}_{\sbfbar}(\cdot)$;  and $\langle H \rangle_{ \rho_{\sbfbar}(\thetabf)}=\Tr(H\rho_{\sbfbar}(\thetabf))$ is the expected observable under the $\sbfbar$th quantum state. 

Evaluating \eqref{eq:H_qem} becomes practically infeasible as the number $n$ of qubits or the number $D$ of unitary maps grows, since the number of terms in the sum \eqref{eq:QPR_overallunitary} grow as $4^{nD}$.
\begin{algorithm}
    \caption{$\mathrm{QEM}$ via $\mathrm{MC}$ sampling}
  \begin{algorithmic}[1]
    \INPUT Ideal unitary gates $\{\Uscr_d^{\theta_d}\}_{d=1}^D$, set of implementable operations $\Omega$, and number of circuits $N_c$
    \OUTPUT $\{\sgn(q_{\sbfbar_l})\}_{l=1}^{N_c}$, $\{\Oscr_{\sbfbar_l}^{\thetabf}\}_{l=1}^{N_c}$ and $\langle H \rangle_{\rho_{\sbfbar_{1:N_c}}(\thetabf)}$
    \STATE Using \eqref{eq:decomposition} obtain the quasi-probability vector $q_d=[q_{d,s}]_{s=0}^{4^n-1}$, for $d=1,\hdots,D$.
    \STATE Compute $Z_d$ and $Z$ using \eqref{eq:Z},  and $p_{d,\sbar}$ using \eqref{eq:pmf}
    \FOR{$l=1,\hdots,N_c$}
      \FOR{$d=1,\hdots,D$}
      \STATE Sample a noisy circuit $s_d \sim p_{d,s}$
      to chose the parameterized gate $\Oscr^{\theta_d}_{s_d}$
      \ENDFOR
      \STATE Compute $\sgn(q_{\sbfbar_l})=\prod_{d=1}^D \sgn(q_{d,\sbar_d})$
\STATE Compute $\langle H_{\sgn} \rangle_{\rho_{\sbfbar_l}}=\sgn(q_{\sbfbar_l}) \Tr(H {\Oscr}^{\thetabf}_{\sbfbar_l}(\rho_0))$
\ENDFOR
\STATE \textbf{Return} $\{\sgn(q_{\sbfbar_l})\}_{l=1}^{N_c}$, $\{\Oscr_{\sbfbar_l}^{\thetabf}\}_{l=1}^{N_c}$ and $\langle H \rangle_{\rho_{\sbfbar_{1:N_c}}(\thetabf)}=\frac{Z}{N_c} \sum_l \langle H_{\sgn} \rangle_{\rho_{\sbfbar_l}} $
  \end{algorithmic}
\end{algorithm}
Consequently, in practice, quasi-probabilistic error mitigation is implemented via a Monte Carlo sampling of $N_c$ circuits $\{\sbfbar_1,\hdots, \sbfbar_{N_c}\}:=\sbfbar_{1:N_c}$ in an i.i.d manner from the distribution $p_{\sbfbar}$ \cite{temme2017error,xiong2022accuracy}. This in turn, gives an unbiased estimator of the QEM-mitigated expected observable  \eqref{eq:H_qem} as
\begin{align}
 \langle H \rangle_{\rho_{\sbfbar_{1:N_c}}(\thetabf)}=Z \sum_{l=1}^{N_c}\frac{1}{N_c} \sgn(q_{\sbfbar_l})\langle H \rangle_{\rho_{\sbfbar_l}(\thetabf)}. \label{eq:H_qem_MC}
\end{align}
We refer to Algorithm 1 for a summary of QEM.

\subsection{QEM-Based Gradient Estimator}
In this section, we describe the QEM-based gradient estimator used to mitigate the bias due to gate noise. At each iteration $t$,  the QEM-based gradient estimator $\hat{g}_t^{\QEM}$ is obtained in two steps. In the first step, the gradient estimator samples $N_c$ noisy circuits $\{\Oscr^{\thetabf^{t-1}}_{\sbfbar_l}\}_{l=1}^{N_c}$, with the current parameter vector $\thetabf^{t-1}$. This is done by using the distribution $p_{\sbfbar}$ as described in Algorithm~1.   
For each $l$th sampled noisy circuit, the gradient estimator then applies the parameter-shift rule as
\begin{align}
    [\hat{g}_{t,l}^{\QEM}]_d=\frac{1}{2}\biggl(\widehat{\langle H \rangle}_{\rho_{\sbfbar_l}(\thetabf^{t-1}+\frac{\pi}{2}{e}_d)} -\widehat{\langle H \rangle}_{\rho_{\sbfbar_l}(\thetabf^{t-1}-\frac{\pi}{2}{e}_d)} \biggr), \label{eq:percircuit_gradient}
    \end{align}with $l=1,\hdots, N_c$. In \eqref{eq:percircuit_gradient}, the term $\widehat{\langle H \rangle}_{\rho_{\sbfbar}(\thetabf)}$ denotes an empirical estimator of the expected observable $\langle H \rangle_{\rho_{\sbfbar}(\thetabf)}$,  obtained from $N_{\QEM}$ i.i.d. measurements $\{H^{\sbfbar}_1,\hdots, H^{\sbfbar}_{N_{\QEM}}\}$ of observable $H$ under the sampled noisy quantum state $\rho_{\sbfbar}(\thetabf)$. This yields the output $H^{\sbfbar}_j=h_y$ with probability
\begin{align}
    p_{\sbfbar}(y|\thetabf) = \Tr(\Pi_y \rho_{\sbfbar}(\thetabf)). \label{eq:circuit_distribution}
\end{align} Finally, the
 $d$th component of the gradient estimator, $[\hat{g}_t^{\QEM}]_d$ is obtained as 
\begin{align}
 [\hat{g}^{\QEM}_t]_d=\frac{Z}{ N_c}&\sum_{l=1}^{N_c}\sgn(q_{\sbfbar_l}) [\hat{g}_{t,l}^{\QEM}]_d\label{eq:QEM_gradientestimator},
\end{align} by averaging the product of the per-circuit gradient estimation $[\hat{g}_{t,l}^{\QEM}]_d$ in \eqref{eq:percircuit_gradient} over the $N_c$ sampled circuits, and multiplying by the normalizing constant $Z$. The QEM-based gradient estimator is described in Algorithm~2.

To enable ease of comparison of the convergence behavior of the QEM-based gradient estimator with the gradient estimators described in previous sections, we fix the measurement budget to $N_m$ shots for the total of $N_c$ circuits. In other words, we fix $N_{\QEM}=N_m/N_c$ number of measurement shots on each sampled circuit.
\begin{algorithm}
    \caption{QEM-Based SGD}
  \begin{algorithmic}[1]
    \INPUT Initialization $\thetabf^{0}$, number of iterations $T$, number of circuit samples $N_c$, learning rate $\{\eta_t\}$
    \OUTPUT Final iterate $\thetabf^T$
    \WHILE{$t\leq T$}
    \STATE Set $\thetabf=\thetabf^{t-1}$
    \STATE Get $\{\sgn(q_{\sbfbar_l})\}_{l=1}^{N_c}, \{\Oscr_{\sbfbar_l}^{\thetabf}\}_{l=1}^{N_c}$ from Algorithm~1
    \FOR{$d=1,\hdots,D$}
    \STATE Set $g_d=[\cdot ]$
      \FOR{$l=1,\hdots,N_c$}
      \STATE  Implement  $ \Oscr_{\sbfbar_l}^{\thetabf \pm  \frac{\pi  }{2}\mathbf{e}_d} (\rho_0) $
      \STATE Compute $[\hat{g}_{t,l}^{\QEM}]_d$ using \eqref{eq:percircuit_gradient}
      \STATE Set
      $g_d[l] = \sgn(q_{\sbfbar_l}) [\hat{g}_{t,l}^{\QEM}]_d$
      \ENDFOR
\STATE Compute 
$[\hat{g}_t^{\QEM}]_d=\frac{Z}{N_c} \sum_{l=1}^{N_c} g_d[l]$
\ENDFOR
\STATE Update parameter vector as
$\thetabf^t \leftarrow \thetabf -\eta_t \hat{g}_t^{\QEM} $
\STATE Update $t\leftarrow t+1$
\ENDWHILE
\STATE Return final parameter iterate $\thetabf^T$
  \end{algorithmic}
\end{algorithm}
\subsection{Properties of the Gradient Estimate}\label{sec:QEM_properties}
The QEM-based gradient estimator of \eqref{eq:QEM_gradientestimator} has two sources of randomness, namely the randomness arising from sampling noisy circuits $\{\Oscr^{\thetabf}_{\sbfbar_l}\}_{l=1}^{N_c}$, as well as that arising from the quantum measurements. Accordingly, we can decompose the QEM-based gradient estimator in \eqref{eq:QEM_gradientestimator} as
\begin{align}
 \hat{g}^{\QEM}_t= g_t+  \xi^{\QEM}_t,   
\end{align} where $\xi^{\QEM}_t$ denotes the total noise due to shot noise and circuit sampling noise, that satisfies
$\Ebb[\xi^{\QEM}_t]=0$ and $\mathrm{var}(\xi^{\QEM}_t)=\Ebb[\lVert \hat{g}^{\QEM}_t- g_t \rVert^2]$, with the expectation taken over the randomly sampled noisy circuits as well as over  quantum measurements. It can be seen that the following equality holds (see proof in Appendix~\ref{app:variance_QEM}) \begin{align}
\mathrm{var}(\xi^{\QEM}_t)= \Ebb[\lVert \hat{g}_t^{\QEM}-g_t^{\Circ}\rVert^2] + \Ebb[\lVert g_t^{\Circ}-g_t \rVert^2] \label{eq:variance_decomposition}, \end{align} where 
\begin{align}
 [g^{\Circ}_t]_d&=\frac{Z}{2 N_c}\sum_{l=1}^{N_c}\sgn(q_{\sbfbar_l}) \Bigl({\langle H \rangle}_{\rho_{\sbfbar_l}(\thetabf^t+\frac{\pi}{2}{e}_d)} \non \\&-{\langle H \rangle}_{\rho_{\sbfbar_l}(\thetabf^t-\frac{\pi}{2}{e}_d)} \Bigr) \label{eq:auxiliary_loss}
 \end{align}is the shot-noise free estimate of the gradient obtained via $N_c$ noisy sampled circuits. While the first term of \eqref{eq:variance_decomposition} captures the impact of finite number, $N_{\QEM}=N_m/N_c$, of measurements made per sampled noisy circuit,  the second term captures the impact of sampling finite number of noisy circuits on the variance of the QEM-based gradient estimator.  
With this insight, the next theorem provides bounds on the variance of the QEM-based gradient estimator.
\begin{theorem}\label{lem:variance_QEM}
The following inequality holds for the variance of the QEM-based gradient estimator,
\begin{align}
   c_1(\gamma) \mathrm{var}(\xi^{\Escr}_t) &\leq  \mathrm{var}(\xi^{\QEM}_t) \leq V^{\QEM},  \label{eq:variance_QEM}
   \end{align}
  where $\mathrm{var}(\xi^{\Escr}_t)$ defined as in \eqref{eq:variance_noisy}; we have $c_1(\gamma)\geq 0$; and \begin{align} V^{\QEM} &= \frac{N_h DZ^2 \Tr(H^2)}{2 N_m}\sup_{\thetabf \in \Real^D ,y \in \{1,\hdots,N_h\}}\Ebb_{\sbfbar}[\nu(p_{\sbfbar}(y|\thetabf)) ] \non\\ & +\frac{Z^2 D \lVert H \rVert_{\infty}^2}{N_c} \label{eq:VQEM} \end{align}
 satisfies the inequality \begin{align}
     V^{\QEM} \geq c_1(\gamma) V^{\Escr}+\frac{c_2(\gamma) D\lVert H \rVert_{\infty}^2}{N_c}, \label{eq:relation_3}
\end{align} where $c_2(\gamma)\geq 1$ and $p_{\sbfbar}(y|\thetabf)$ is defined as in \eqref{eq:circuit_distribution}. Furthermore, in the special case of the depolarizing channel, the inequality $c_1(\gamma)\geq 1$ is satisfied,  and both $c_1(\gamma)$ and $c_2(\gamma)$ are non-decreasing functions of $\gamma$ with $c_1(0)=c_2(0)=1$.
\end{theorem}
\begin{proof}
Proof can be found in Appendix~\ref{app:variance_QEM}.
\end{proof}

Theorem~\ref{lem:variance_QEM} highlights three key points. First, as can be seen from \eqref{eq:variance_QEM}, the variance of the QEM-based gradient estimator is at least $c_1(\gamma)$ times the variance of the unmitigated gradient estimator. For depolarizing noise, the function  $c_1(\gamma)\geq 1$ is non-decreasing with $\gamma$, implying that QEM-based gradient estimator has larger variance than the biased noisy gradient estimator.

Second, the upper bound $V^{\QEM}$ in \eqref{eq:VQEM}  on the variance of the QEM-based gradient estimator comprises the contribution of shot noise, as captured by the first term, as well as of the circuit sampling noise, as captured by the second term. Accordingly, the bound \eqref{eq:VQEM} suggests that, even when an infinite number of circuits is sampled, \ie when $N_c \rightarrow \infty$, the finite number of per-circuit measurements results in a non-zero variance. Finally, the inequality \eqref{eq:relation_3} relates the bound $V^{\QEM}$ of the variance with QEM to the bound $V^{\Escr}$ obtained without QEM, as defined in \eqref{eq:variance_noisy}. Noting that functions $c_1(\gamma)$ and $c_2(\gamma)$ are non-decreasing with $\gamma$ for depolarizing noise, the inequality in \eqref{eq:relation_3} suggests that the variance of the QEM-based estimator generally increases with gate noise.
\subsection{Convergence of QEM-Based SGD}
Based on the analysis of the stochastic gradient in the previous subsection, we have the following convergence result for SGD using the QEM-based gradient estimator $\hat{g}_t^{\QEM}$.
\begin{theorem}\label{thm:QEM}
Under Assumption~\ref{assum:1}, for any given initial point $\theta^0$, the following bound on the optimality gap holds for the SGD employing QEM-based gradient estimator in \eqref{eq:QEM_gradientestimator}, given any fixed learning rate $\eta_t=\eta\leq 1/\Lscr$
\begin{align}
    \Ebb[L(\theta^T)]-L^{*} &\leq (1-\eta \mu)^T (\Ebb[L(\theta^0)]-L^{*})\non \\&+\frac{1}{2} \biggl[\frac{\eta \Lscr V^{\QEM}}{\mu} \biggr]\label{eq:convergence_2},
\end{align}where the expectations are taken over quantum measurements and noisy circuit samples, and $V^{\QEM}$ is defined as in \eqref{eq:variance_QEM}.
Furthermore, given some target error level $\delta>0$, for learning rate $\eta =\eta^{\QEM} \leq \min \{ \frac{1}{\Lscr}, \frac{\delta \mu}{\Lscr V^{\QEM}}\}$, a number of iterations
\begin{align}
   T^{\QEM}=\tilde{\Oscr}\biggl(\log \frac{1}{\delta} + \frac{V^{\QEM}}{\delta \mu} \biggr)\frac{\Lscr}{\mu} \label{eq:iterations_QEM}
\end{align} is sufficient to ensure the error $\Ebb[L(\thetabf^T)]-L^{*}=\Oscr(\delta)$.
\end{theorem}

Theorem~\ref{thm:QEM} can be used to compare the convergence of QEM-based SGD with that of the SGD under shot and gate noise in Theorem~\ref{thm:noisy_convergence}. While the presence of quantum gate noise forces SGD to settle at an error floor of $\Oscr(\delta)$ with $\delta> B^{\Escr}\mu$, QEM can achieve any error floor, $\Oscr(\delta)$ for $\delta>0$, in a finite number $T^{\QEM}$ of iterations. We refer to Sec. I for further discussion on this result.
 

\section{Experiments}\label{sec:experiments}
In this section, we present numerical results concerning the VQE to solve a weighted max-cut combinatorial optimization problem.
\subsection{Weighted Max-Cut Hamiltonian}
Consider an undirected graph $\Gscr=(V,E)$ with $V=\{1,\hdots,n\}$ denoting the set of vertices of the graph and $E \subseteq V \times V$ denoting the set of edges. Each edge $(i,j) \in E$ has an associated weight $w_{i,j}>0$ such that $w_{i,j}=w_{j,i}$. A cut of the graph defines a partition of the vertices into two distinct subsets. Specifically, a cut assigns a binary variable $x_v \in \{0,1\}$ to each vertex $v \in V$ depending on whether it belongs to one subset or to the other.   In the weighted max-cut optimization problem, the goal is to find the cut that maximizes the sum of the weights of the edges that connect the vertices belonging to the two distinct subsets, \ie, of the edges crossing the cut. This corresponds to maximizing the  cost function
\begin{align}
    C(\xbf)= \sum_{(i,j) \in E}w_{i,j}x_{i}(1-x_j) +\sum_{i}w_{i,i} x_i \label{eq:maxcut_objective}
\end{align}over the binary vector $\xbf=(x_1,\hdots,x_n)$ with $x_i \in \{0,1\}$ for all $i=1,\hdots,n$.  Note that \eqref{eq:maxcut_objective} also imposes an additional penalty on the self-weights assigned to each vertices via the second summation.

The objective function in \eqref{eq:maxcut_objective} can be converted to an Ising Hamiltonian via the mapping $x_i \mapsto (1-Z_i)/2$ where $Z_i=(I \otimes \hdots I \otimes Z \otimes I \otimes \hdots \otimes I)$ denotes an $n$-qubit operator that applies a Pauli-$Z$ gate on the $i$th-qubit. The weighted max-cut problem of maximizing  the cost in \eqref{eq:maxcut_objective} over $\xbf$ can be then equivalently expressed as minimizing the expected value $\langle H \rangle_{\vert \Psi(\thetabf) \rangle}$ of the Ising Hamiltonian
\begin{align}
    H =\sum_{i=1}^n w_{i,i} Z_i +\sum_{i<j} w_{i,j}Z_i Z_j \label{eq:problem_Hamiltonian}
\end{align} over the quantum state $\vert \Psi(\thetabf) \rangle$. To produce the state $\vert \Psi(\thetabf) \rangle = U(\thetabf) \vert 0 \rangle$,
 we consider the hardware-efficient ansatz adopted in \cite{amaro2022filtering} as shown in Figure~\ref{fig:ansatz}. Accordingly, the PQC comprises of single-qubit Pauli $Y$-rotation gates $R_y(\theta)$, as well as two-qubit CNOT gates. The dashed box in Figure~\ref{fig:ansatz} represent a layer of the PQC which can be repeated to obtain deep variational quantum circuits.

\begin{figure}
    \centering
    \includegraphics[scale=0.5,clip=true, trim= 2.6in 3.1in 2.3in 1.4in]{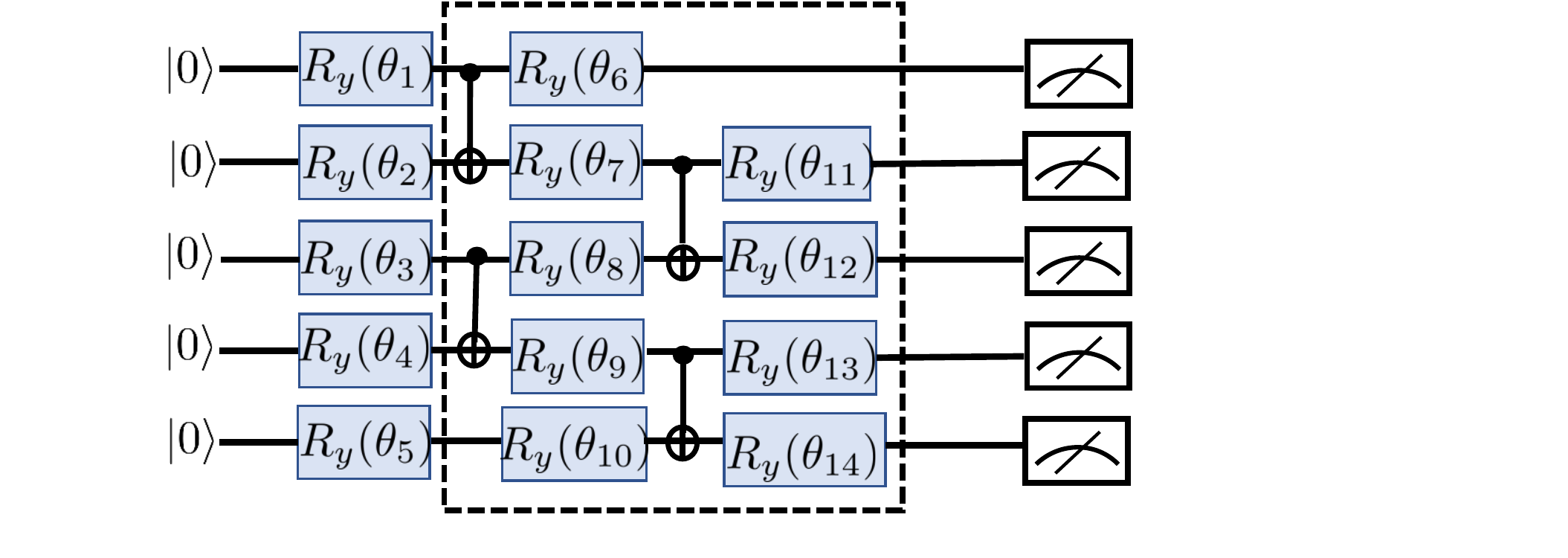}
    \caption{Hardware-efficient ansatz used in the experiments. The gates in the dashed box are repeated $L$ times.}
    \label{fig:ansatz}
\end{figure}
\subsection{Results}
We first consider a simple complete graph $\Gscr$ with $n=3$ vertices and a random $3 \times 3$ weight matrix $w=[0.41, 0.44, 0.55; 0.44, 0.97, 0.22; 0.55, 0.22,0.89]$ whose $(i,j)$th element $w_{i,j}$ denotes the weight of the edge $(i,j) \in E$. We consider the PQC defined by the hardware-efficient ansatz in Figure~\ref{fig:ansatz} with a single-layer, \ie, with $L=1$. We assume that depolarizing noise, with error probability $\epsilon$, acts only on the CNOT gates. This is motivated by the experimental observations using Qiskit that noisy CNOT gates yield larger error in measurement outcomes as compared to noisy single-qubit gates.
\begin{figure*}[h!]
     \begin{minipage}[c]{0.45 \linewidth}
     \centering
      \includegraphics[scale=0.5,clip=true, trim=0in 0in 0in 0.45in]{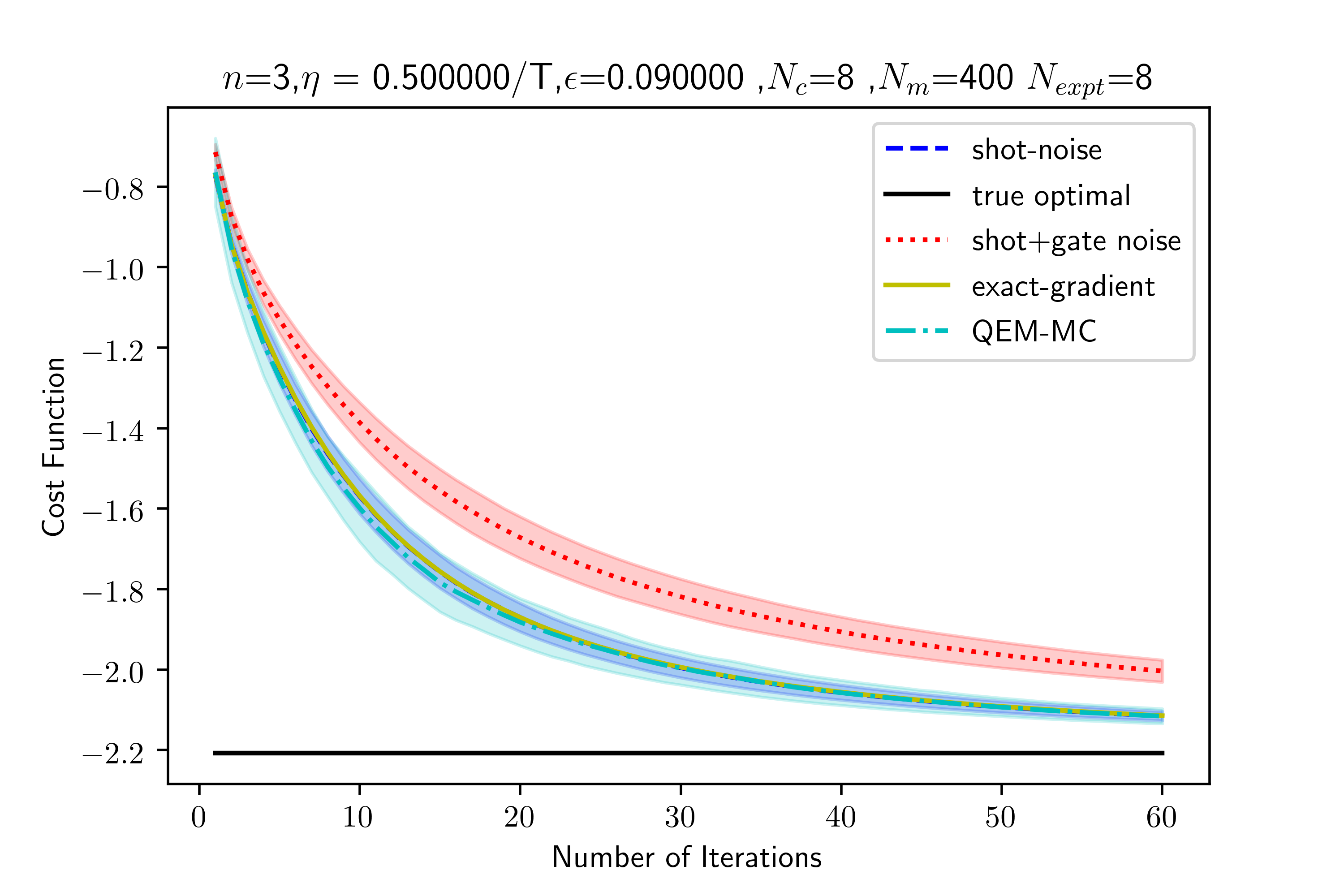}
      \end{minipage} \hfill
      \begin{minipage}[c]{0.45 \linewidth}
      \centering
     \includegraphics[scale=0.5,clip=true, trim=0in 0in 0in 0.45in]{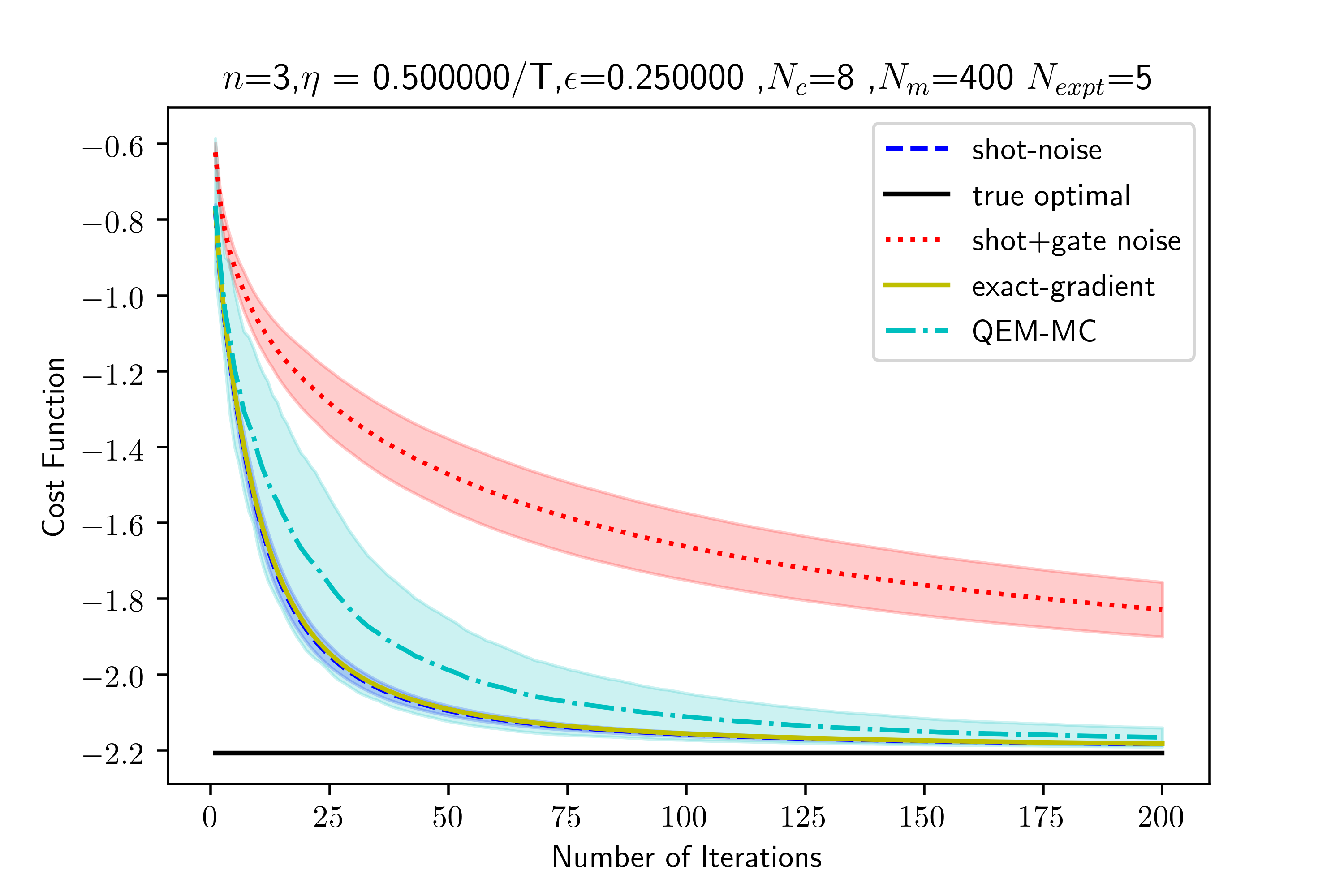}
     \end{minipage}
         \caption{Convergence analysis of the SGD -- when exact gradients can be computed; when only shot-noise is present; when shot and gate noise are present; and when QEM is employed --  as a function of the number of iterations for $n=3$ qubit system subjected to depolarizing noise on the CNOT gates. The learning rate  is $\eta_t = 0.5/t$; $N_c=8$; $N_m=400$; and (left) $\epsilon=0.09$ and (right) $\epsilon=0.25$.}
         \label{fig:low_error}
\end{figure*}
   
In Figure~\ref{fig:low_error}, we study the convergence of SGD  when exact gradients can be computed; when only shot noise is present, with number of measurements $N_m=400$; when both shot and gate noise are present, with the latter modelled by the mentioned depolarizing noise with error probability $\epsilon$; and when QEM is used to combat gate noise with number of sampled noisy circuits $N_c=8$.  We compare the figure on the  left, corresponding to a smaller noise of $\epsilon=0.09$, with that on the right with a higher noise of $\epsilon=0.25$. We keep the noise level relatively high to achieve a better contrast between the above three scenarios under study, given the robustness of VQAs to noise \cite{sharma2020noise}. For each of the three scenarios, namely shot noise, shot+gate noise and QEM, we conduct $8$ experiments with the SGD iterations starting from the same fixed initial point $\thetabf^0$ to obtain a set of parameter iterates. In each experiment, the cost $\langle H \rangle_{\vert \Psi(\thetabf^t)\rangle}$ corresponding to the obtained parameter iterate $\thetabf^t$, for $t=1,\hdots,T$, is evaluated exactly (\ie $N_m \rightarrow \infty$). The mean of the costs over the experiments is indicated by the bold curve, while the spread is indicated by the lighter shadows. 

Overall, the experimental findings of Fig.~\ref{fig:low_error} corroborate our theoretical analysis. In particular, by comparing the two figures in Fig.~\ref{fig:low_error}, we observe that a larger  strength of the depolarizing noise causes a larger error floor for SGD in the presence of shot and gate noise. In contrast, QEM  achieves a lower error floor under both low and high noise strengths. When the noise strength is smaller and a sufficient number of circuits are sampled as in Fig.~\ref{fig:low_error}(left), the variance of the QEM-based estimator is reduced (see Table~\ref{tab:results}), thereby ensuring convergence at a rate comparable to the shot-noise only case. In contrast, when the noise strength is larger, by Theorem~\ref{lem:variance_QEM},  QEM requires more circuits to be sampled per iteration.

 \begin{figure*}
     \centering
     \begin{minipage}[c]{0.45 \linewidth}
     \centering
     \includegraphics[scale=0.6,clip=true, trim = 0in 0in 0in 0in]{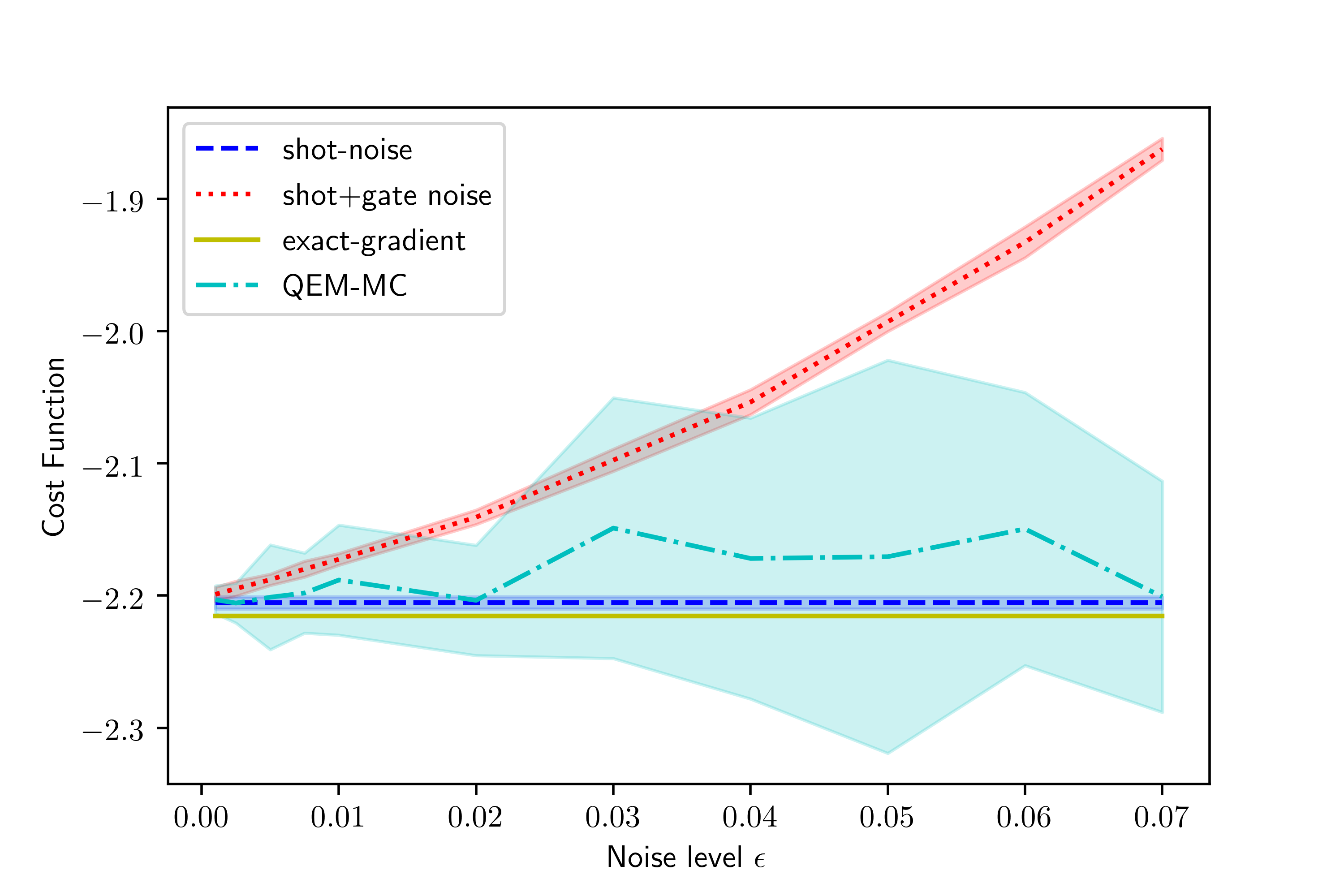}
     \end{minipage} \hfill\begin{minipage}[c]{0.45 \linewidth}
     \centering
      \includegraphics[scale=0.6,clip=true, trim = 0in 0in 0in 0in]{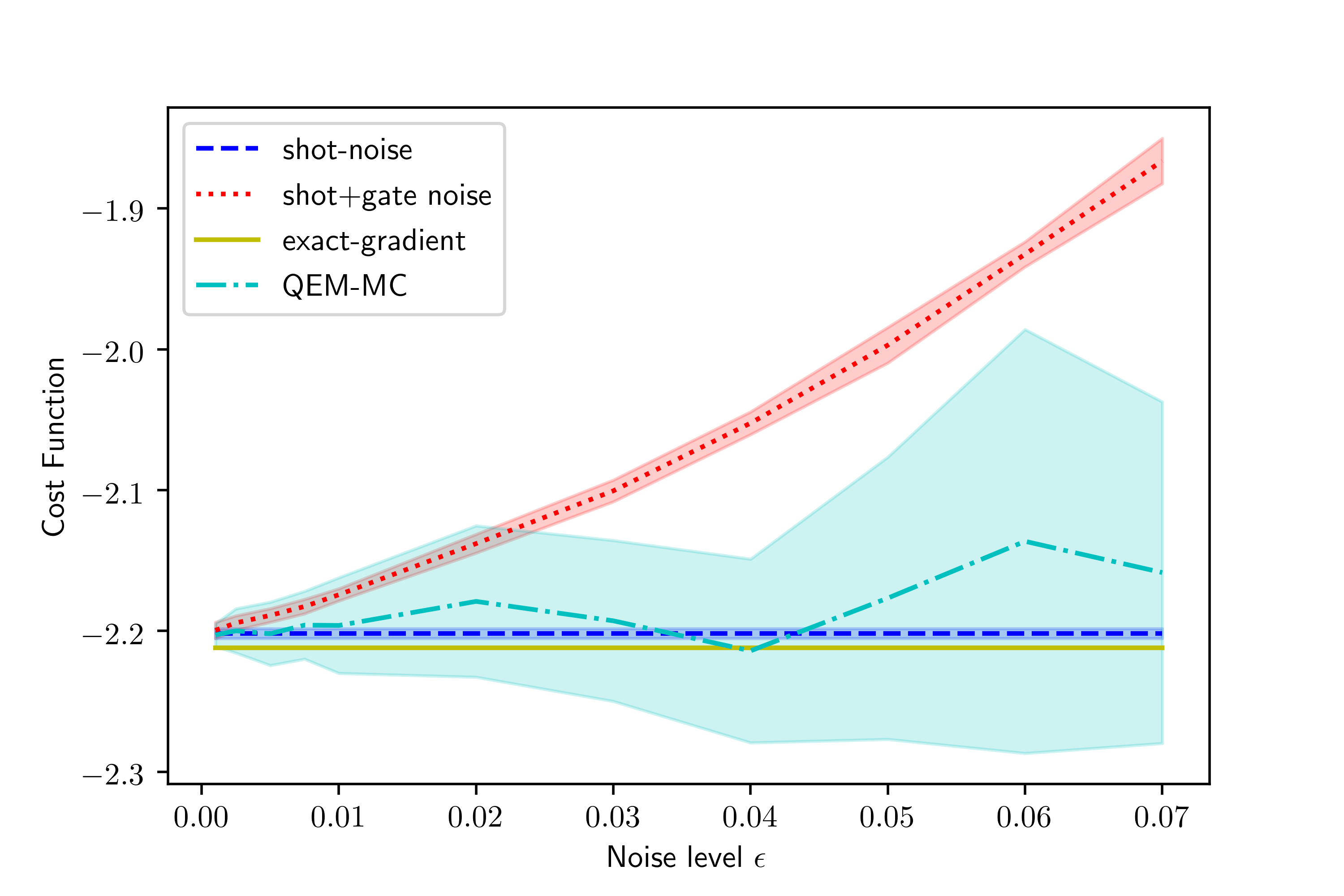}
      \end{minipage}
     \caption{Cost function $L(\thetabf^T)$ after $T=10$ iterations of the SGD when only shot noise is present; when both shot and gate noise are present; and when QEM is employed, as a function of the noise level $\epsilon$. (Left) Circuit samples $N_c=7$ and (right) $N_c=10$. Other parameters are set as $n=5$, $L=1$, $N_m=10240$, and $\eta_t=0.14$. Ground state eigenvalue for the max-cut problem is $-3.24$.}
     \label{fig:errorprobability}
 \end{figure*}
 In Figure~\ref{fig:errorprobability}, we study the exact loss $L(\thetabf^T)$ (evaluated with $N_m \rightarrow \infty$) of the SGD iterate after $T=10$ iterations when exact gradient can be computed; when shot-noise is present, with $N_m=10240$ measurement shots; when shot and gate noise are present; and when QEM is employed to mitigate the bias, as a function of the increasing noise level $\epsilon$. The figure in the left corresponds to smaller number, $N_c=7$, of circuit samples, while the figure on the right corresponds to $N_c=10$. The max-cut problem has $n=5$ vertices and a random $5 \times 5$ weight matrix $w=[ 0.42, 0.43,0.55,0.96,0.22;$ $ 0.44, 0.89, 0.07, 0.87,0.01;$ $ 0.55,$ $0.07,0.77, 0.18, 0.15;$ $ 0.96,$ $ 0.87,0.18,0.77,0.51;$ $ 0.22,$ $0.01, 0.15,0.51,0.84]$. As before, we use the harware-efficient ansatz in Fig.~\ref{fig:ansatz}, with $L=1$ layer and depolarizing noise assumed to act only on CNOT gates. We fix the learning rate as $\eta_t=0.14$. 
 
  It can be seen from Fig.~\ref{fig:errorprobability} that the presence of gate noise induces a significant bias as the strength of the noise increases. For  large noise level, QEM successfully lowers the bias induced by gate noise. However, when the noise strength is small, comparing the figures in the left and right shows that the higher variance of the QEM, resulting due to small $N_c$, may offset the benefit of QEM.
  
  \begin{figure*}
     \centering
     \begin{minipage}[c]{0.45 \linewidth}
     \centering
     \includegraphics[scale=0.6,clip=true, trim = 0in 0in 0in 0in]{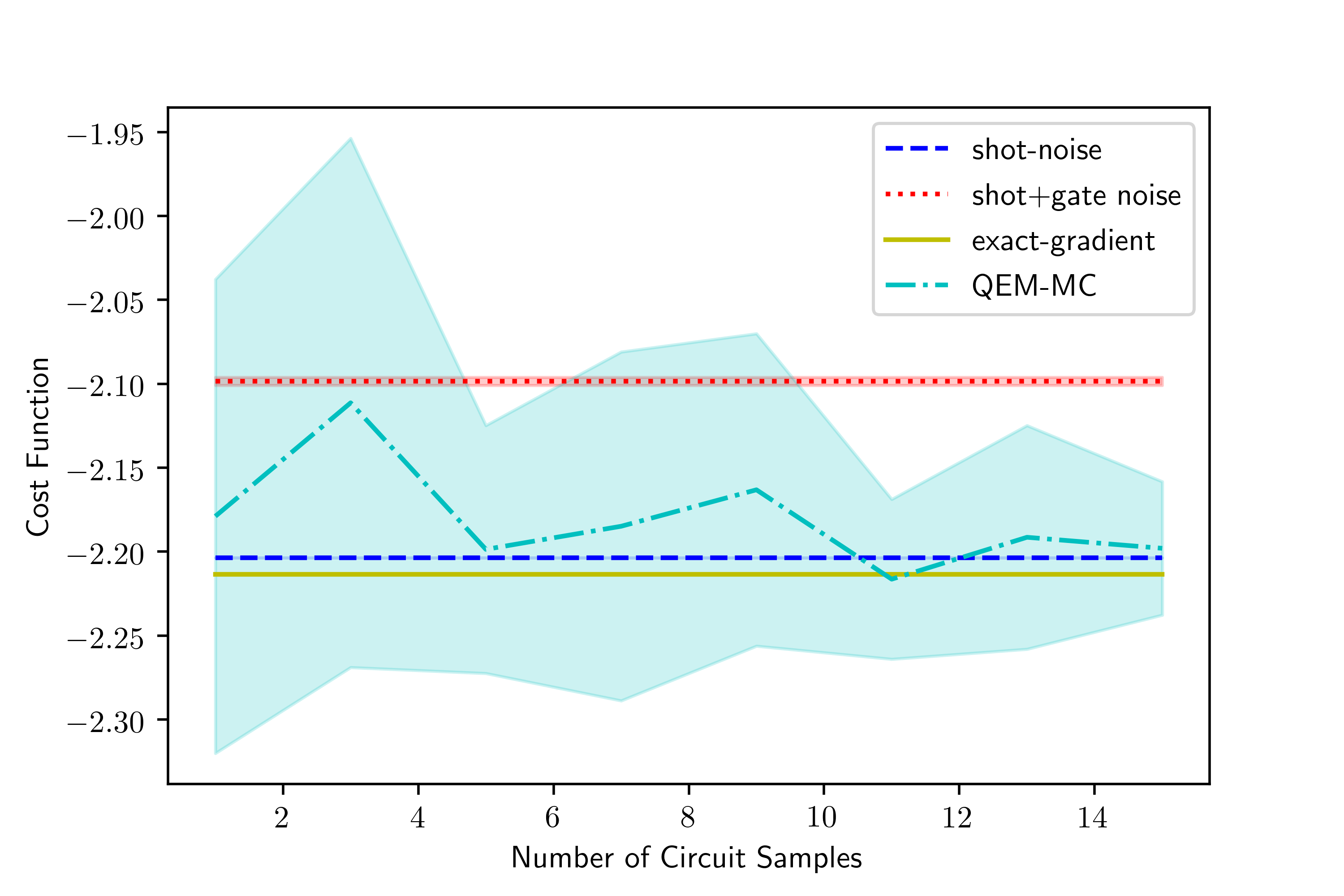}
     \end{minipage} \hfill\begin{minipage}[c]{0.45 \linewidth}
     \centering
      \includegraphics[scale=0.6,clip=true, trim = 0in 0in 0in 0in]{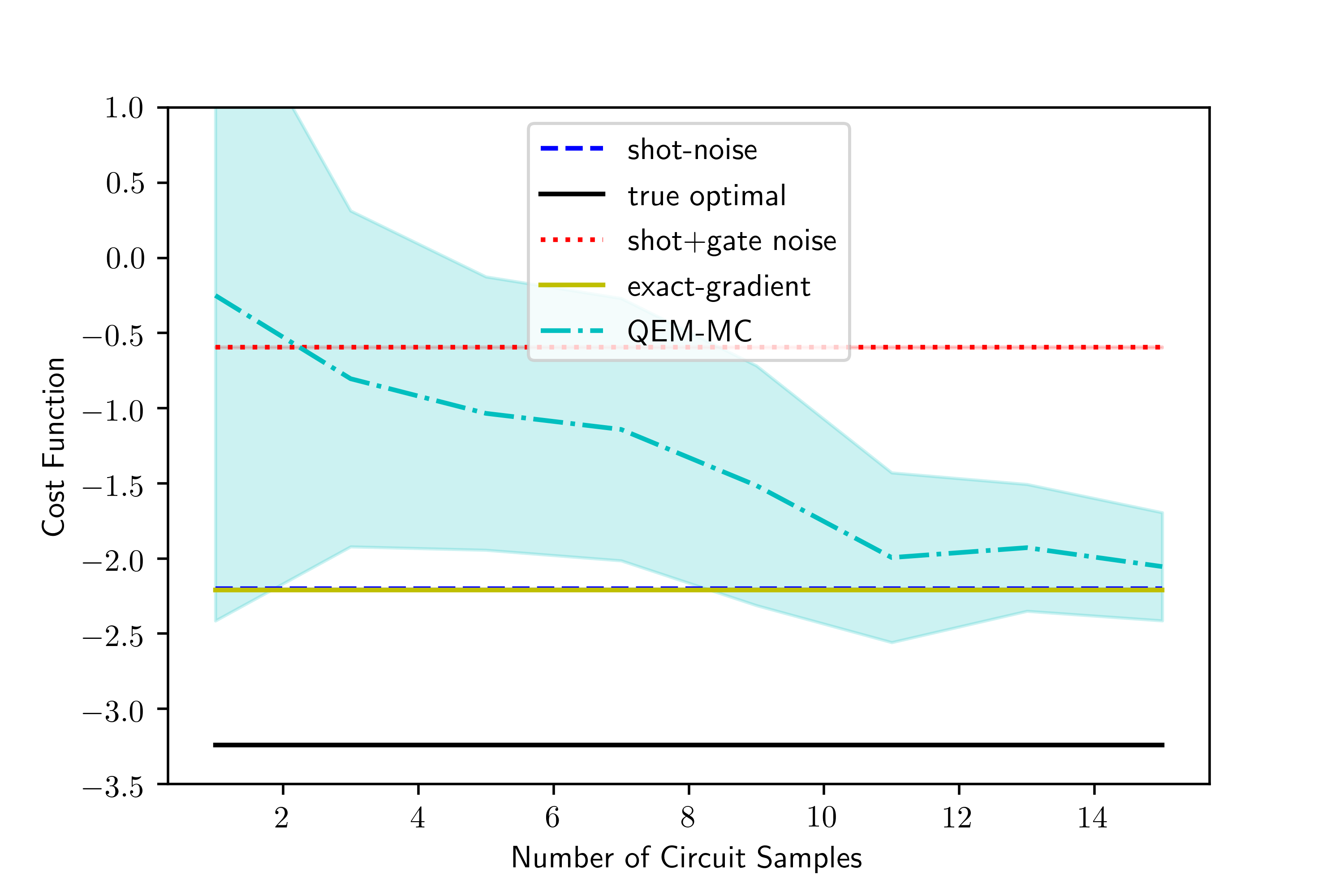}
      \end{minipage}
     \caption{Cost function $L(\thetabf^T)$ after $T=10$ iterations of the SGD when only shot noise is present; when both shot and gate noise are present; and when QEM is employed, as a function of the number of circuit samples $N_c$. (Left) low noise level $\epsilon=0.03$ and (Right) high noise level ($\epsilon=0.25$). Other parameters are set as $n=5$, $L=1$, $N_m=10240$, and $\eta_t=0.14$. Ground state eigenvalue for the max-cut problem is $-3.24$.}
     \label{fig:circuitsamples}
 \end{figure*}
  Finally, Fig.~\ref{fig:circuitsamples} demonstrates the impact of increasing number of circuit samples on the loss $L(\thetabf^T)$ of the SGD iterate after $T=10$ iterations. The experimental setting is same as in Fig.~\ref{fig:errorprobability}. As can be seen from the figure, increasing the number of circuit samples significantly reduces the variance of the QEM and ensure performance close to the SGD with perfect error mitigation (blue curve).
 \section{Conclusions}
This paper seeks an answer to the following question: Can quasi-probabilistic error mitigation (QEM) be beneficial in improving the convergence of SGD for the implementation of VQEs in NISQ devices? By analyzing the convergence properties of SGD for VQEs, we have shown that the quantum gate noise inherent in the operation of NISQ devices results in a non-zero error floor. Achieving lower error levels requires the use of QEM.  However, for larger error levels, the increase in the variance of the stochastic gradient estimate caused by QEM may entail the need for a larger number of SGD iterations when QEM is deployed. In particular,  when the noise strength is high and insufficient time is available at each SGD iteration to sample circuits for QEM, QEM may require a larger number of iterations to converge to the desired error level. Conversely, when the number of circuits sampled per iteration is sufficiently large, QEM can yield a significant reduction in the number of SGD iterations as compared to a conventional system without error mitigation.
\appendices
\section{Proof of Lemma~3.1}\label{app:1}
To prove Lemma~3.1, we fix the iteration index $t$, which is dropped from the notation. 
Let
 $X_{d,\pm}=\widehat{ \langle H\rangle}_{\vert \Psi(\thetabf\pm\frac{\pi}{2}{e}_d) \rangle}-  \langle H\rangle_{\vert \Psi(\thetabf\pm\frac{\pi}{2}{e}_d)\rangle}$ denote the difference between the estimated and true expectation of the observable $H$ under the quantum state $\vert \Psi(\thetabf\pm\frac{\pi}{2}{e}_d)\rangle$ whose $d$th parameter is phase-shifted by $\pm\pi/2$.  For brevity, we also henceforth use the notation  $\vert \Psi_{d,\pm}\rangle=\vert \Psi(\thetabf\pm {e}_d \frac{\pi}{2})\rangle$. The variance of the gradient estimate \eqref{eq:gradient_estimate} can be written as
\begin{align}
  \mathrm{var}(\xi)&=\Ebb\biggl[\sum_{d=1}^{D}\biggl(\frac{1}{2}(\widehat{ \langle H\rangle}_{\vert \Psi_{d,+}\rangle}-\widehat{ \langle H\rangle}_{\vert \Psi_{d,-}\rangle}) \non \\&-\frac{1}{2}( \langle H\rangle_{\vert \Psi_{d,+}\rangle}- \langle H\rangle _{\vert \Psi_{d,-}\rangle})\biggr)^2 \biggr] \non \end{align} \begin{align}
  &=  \sum_{d=1}^{D} \frac{1}{4} \Ebb \biggl[(X_{d,+}-X_{d,-})^2 \biggr] \non \\&=  \sum_{d=1}^{D}\frac{1}{4} \Bigl(\Ebb[X_{d,+}^2]+\Ebb[X_{d,-}^2 ] \Bigr), \label{eq:1_11}
\end{align}where the expectation is with respect to the $N_m$  measurements of the quantum state $\vert \Psi(\thetabf + e_d\frac{\pi}{2}) \rangle$ and the $N_m$ measurements of the quantum state $\vert \Psi(\thetabf - e_d\frac{\pi}{2}) \rangle$ for $d=1,\hdots,D$. The random variables  $X_{d,+}$ and $X_{d,-}$ are thus independent for $d=1,\hdots,D$, which results in the equality in \eqref{eq:1_11}.

The expectation 
$
    \Ebb[X_{d,+}^2]$ is equal to the variance $\mathrm{var}(\widehat{\langle H \rangle}_{\vert \Psi_{d,+} \rangle})
$ of the random variable $\widehat{\langle H \rangle}_{\vert \Psi_{d,+} \rangle}$. Let $Y$ be the random variable that defines the index of the  measurement of the observable $H$, and let us write as $H=h_{Y}$ for the corresponding measurement output. We denote as
 $W_{y}=\Ibb\{Y=y\}$, for $y =1,\hdots,N_h$, the Bernoulli random variable determining whether $Y=y$ ($W_y=1)$ or not ($W_y=0$). Noting that the quantum measurements are i.i.d., it follows from the definition of expectation $\widehat{\langle H \rangle}_{\vert \Psi_{d,+} \rangle}$ in \eqref{eq:shot_noise_observable_new} that
  \begin{align}
    \Ebb[X^2_{d,+}]&=\frac{1}{N_m}\var\Bigl(\sum_{y=1}^{N_h}h_y W_y\Bigr) \label{eq:1} \\
    &= \frac{1}{N_m}\Ebb\biggl[\biggl(\sum_{y=1}^{N_h} h_y\Bigl(W_y-p\Bigl(y\Bigl|\theta+e_d \frac{\pi}{2}\Bigr)\Bigr) \biggr)^2 \biggr]\label{eq:acrossproofs_1} \end{align} 
    \begin{align}
    & \stackrel{(a)}{\leq} \frac{1}{N_m}\Bigl(\sum_{y=1}^{N_h} h_y^2\Bigr) \sum_{y=1}^{N_h}\var(W_y)\\
    &\stackrel{(b)}{=}\frac{1}{N_m}\Bigl(\sum_{y=1}^{N_h} h_y^2\Bigr) \sum_{y=1}^{N_h}\nu\Bigl(p\Bigl(y\Bigl|\thetabf+e_d \frac{\pi}{2}\Bigr)\Bigr)\non\\
    & \leq \frac{N_h}{N_m}\Bigl(\sum_{y=1}^{N_h} h_y^2\Bigr) \nu = \frac{N_h \Tr(H^2)}{N_m}\nu, \label{eq:1_1}
    \end{align}
  where $(a)$ follows from the Cauchy-Schwarz inequality;  $(b)$ follows since the variance of the Bernoulli random variable $W_{y}$
  can be computed as
   \begin{align*}
       \mathrm{var}(W_{y})&=\Ebb[W_{y}^2]-(\Ebb[W_{y}])^2=\nu\Bigl(p\Bigl(y \Bigl|\thetabf+e_d \frac{\pi}{2}\Bigr)\Bigr),
   \end{align*}  with  $\nu(x)=x(1-x)$ for $x \in (0,1)$. The last inequality follows from the definition of the quantity $\nu$ in \eqref{eq:nu2}.
   
   In a similar way, it can be shown that the following inequality holds
   \begin{align}
      \Ebb[X^2_{d,-}]&\leq \frac{N_h\Tr(H^2)}{N_m} \nu.
   \end{align}
   Using these in \eqref{eq:1_11} then yields the inequality
   \begin{align}
       \mathrm{var}(\xi) \leq \frac{D N_h \Tr(H^2)\nu}{2N_m},
   \end{align} concluding the proof.
\section{Smoothness of Loss Function $L(\thetabf)$} \label{app:smoothness}
In the following lemma, we show that the loss function $L(\thetabf)$ in \eqref{eq:problem1} is indeed always $\Lscr$-smooth for a constant $\Lscr$ that depends on the number of PQC parameters $D$, and on the observable $H$. In contrast, the loss function $L(\thetabf)$ need not necessarily satisfy the PL-condition.
\begin{lemma}
The loss function $L(\thetabf)$ as defined in \eqref{eq:problem1} is $\Lscr$-smooth with $\Lscr=D^{3/2}\sum_{y=1}^{N_h}|h_y|$.
\end{lemma}
\begin{proof}
To prove the this result, we note that the smoothness condition can be equivalently written as
the inequality $\lVert\nabla^2 L(\theta)\rVert_2 \leq \Lscr$ on the Hessian $\nabla^2 L(\theta)$. We then note the following steps: 
\begin{align}
    &[\nabla^2 L(\thetabf)]_{i,j} = \frac{\partial^2 L(\thetabf)}{\partial \theta_j \partial \theta_i} \\
    &\stackrel{(a)}{=}\frac{1}{4}\Tr \biggl[H \biggl(\Psi\Bigl(\thetabf+{e}_i \frac{\pi}{2}+{e}_j \frac{\pi}{2}\Bigr)-\Psi\Bigl(\thetabf+{e}_i \frac{\pi}{2}-{e}_j \frac{\pi}{2}\Bigr)\non \\&-\Psi\Bigl(\thetabf-{e}_i \frac{\pi}{2}+{e}_j \frac{\pi}{2}\Bigr)+\Psi\Bigl(\thetabf-{e}_i \frac{\pi}{2}-{e}_j \frac{\pi}{2}\Bigr) \biggr) \biggr]\non \\
    &\stackrel{(b)}{\leq } \frac{1}{4}\biggl | \Tr \biggl[H \biggl(\Psi\Bigl(\thetabf+{e}_i \frac{\pi}{2}+{e}_j \frac{\pi}{2}\Bigr)-\Psi\Bigl(\thetabf+{e}_i \frac{\pi}{2}-{e}_j \frac{\pi}{2}\Bigr)\non \\&-\Psi\Bigl(\thetabf-{e}_i \frac{\pi}{2}+{e}_j \frac{\pi}{2}\Bigr)+\Psi\Bigl(\thetabf-{e}_i \frac{\pi}{2}-{e}_j \frac{\pi}{2}\Bigr) \biggr) \biggr] \biggr |\non \\
    & \stackrel{(c)}{\leq} \frac{\lVert H \rVert_{\infty}}{4}  \biggl \lVert \Psi\Bigl(\thetabf+{e}_i \frac{\pi}{2}+{e}_j \frac{\pi}{2}\Bigr)-\Psi\Bigl(\thetabf+{e}_i \frac{\pi}{2}-{e}_j \frac{\pi}{2}\Bigr)\non \\&-\Psi\Bigl(\thetabf-{e}_i \frac{\pi}{2}+{e}_j \frac{\pi}{2}\Bigr)+\Psi\Bigl(\thetabf-{e}_i \frac{\pi}{2}-{e}_j \frac{\pi}{2}\Bigr)\biggr \rVert_1 \end{align}
    \begin{align}
    & \stackrel{(d)}{\leq} \frac{ \lVert H \rVert _{\infty}}{4} \biggl( \biggl \lVert \Psi\Bigl(\thetabf+{e}_i \frac{\pi}{2}+{e}_j \frac{\pi}{2}\Bigr)-\Psi\Bigl(\thetabf+{e}_i \frac{\pi}{2}-{e}_j \frac{\pi}{2}\Bigr)\biggl \rVert_1 \non \\&+ \biggl\lVert-\Psi\Bigl(\thetabf-\mathbf{e}_i \frac{\pi}{2}+\mathbf{e}_j \frac{\pi}{2}\Bigr)+\Psi\Bigl(\thetabf-{e}_i \frac{\pi}{2}-{e}_j \frac{\pi}{2}\Bigr)\biggr\rVert_1\biggr)\non\\
    & \stackrel{(e)}{\leq} \lVert H\rVert_{\infty} \leq \sum_{y=1}^{N_h} |h_y|. \label{eq:new1}
\end{align} The equality in $(a)$ follows from a double application of parameter shift rule \cite{amaro2022filtering}; the inequality in $(b)$ follows since the inequality $x \leq |x|$ holds for $x \in \Real$; $(c)$ follows from tracial matrix H{\"o}lder's inequality \ie $|\Tr(A^{\dag}B)|\leq \lVert A\rVert_{\infty} \lVert B \rVert_1$ \cite{baumgartner2011inequality}; $(d)$ follows from the triangle inequality of the trace norm, $\lVert A \rVert_1=\sqrt{\Tr(A A^{\dag})}$ \cite{wilde2013quantum}; inequality $(e)$ follows since the trace norm of a density matrix is bounded as $0 \leq \lVert \rho \rVert _1 \leq 2$ \cite{wilde2013quantum}; and the last inequality follows from triangle inequality and the inequality $\lVert \Pi_y \rVert_{\infty}\leq 1$, \ie, $\lVert H \rVert_{\infty}=\lVert\sum_{y=1}^{N_h}h_y \Pi_y\rVert_{\infty} \leq \sum_{y=1}^{N_h}|h_y| \lVert\Pi_y \rVert _{\infty} \leq  \sum_{y=1}^{N_h}|h_y|$.
This in turn implies the inequalities
$\lVert \nabla^2 L(\thetabf)\rVert _2 \leq \sqrt{D} \lVert \nabla^2 L(\thetabf)\rVert _{\infty} \leq D^{3/2} \sum_{y=1}^{N_h} |h_y|.$
\end{proof}

In contrast to the smoothness assumption, the convexity of the loss function $L(\thetabf)$ in Assumption~\ref{assum:1} is not always satisfied. It provides a useful working condition for the analysis, which was also adopted in \cite{schuld2019evaluating}, \cite{gentini2020noise}. Convexity may be recovered by applying the optimization approach for PQC detailed in \cite{banchi2020convex}. We leave an investigation of this idea to future works.
\section{Proof of Lemma~4.1}\label{app:biased_estimator}
We first obtain an upper bound on the variance $\mathrm{var}(\xi^{\Escr}_t)=\Ebb[\lVert \hat{g}_t^{\Escr}-g^{\Escr}_t \rVert ^2]$ of the gradient estimator. As in Appendix~\ref{app:1}, we drop the index $t$ and define  $X_{d,\pm}=\widehat{\langle H \rangle}_{\rho^{\Escr}(\thetabf\pm{e}_d\frac{\pi}{2})}-{\langle H \rangle}_{\rho^{\Escr}(\thetabf\pm {e}_d\frac{\pi}{2})}$ as the difference between the estimated and true expectation of the observable $H$ under the noisy quantum state $\rho^{\Escr}(\thetabf \pm \frac{\pi}{2}e_d)$. By \eqref{eq:1_11}, we have the equality
\begin{align}
    \mathrm{var}(\xi^{\Escr})=\sum_{d=1}^{D} \frac{1}{4}\Bigl(\Ebb[X_{d,+}^2]+\Ebb[X_{d,-}^2]\Bigr),\label{eq:2}
\end{align}where $\Ebb[X_{d,\pm}^2]=\mathrm{var}\Bigl(\widehat{\langle H \rangle}_{\rho^{\Escr}(\thetabf\pm {e}_d\frac{\pi}{2})}\Bigr)$. Recall from \eqref{eq:shot_noise_observable_new} the following definition 
\begin{align}
\widehat{\langle H \rangle}_{\rho^{\Escr}(\thetabf + {e}_d\frac{\pi}{2})}=\sum_{y=1}^{N_h}  \frac{h_y}{N_m} \sum_{j=1}^{N_m} \Ibb\{Y^{\Escr}_{j}=y\} \label{eq:QEM_5}
\end{align}of estimated expectation of the observable under noisy state $\rho^{\Escr}(\thetabf + {e}_d\frac{\pi}{2})$  holds, where $Y^{\Escr}_{j}$ is the random variable corresponding to the index of the the $j$th measurement of the observable $H$.
In a manner similar to \eqref{eq:1_1}, it can be shown that the following upper bound holds:
\begin{align}
&\Ebb[X^2_{d,\pm}]
\leq \frac{N_h\Tr(H^2)}{N_m}\sup_{\theta \in \Real^D, y \in \{1,\hdots,N_h\}} \nu(p^{\Escr}(y|\thetabf)), \label{eq:1_4}
\end{align} where $p^{\Escr}(y|\thetabf)$ is defined in \eqref{eq:probability_gatenoise}.  Using \eqref{eq:1_4} in \eqref{eq:2} yields the upper bound \eqref{eq:variance_noisy}.

To obtain the relation in \eqref{eq:relation_1}, we use the decomposition \eqref{eq:decomposition_density} of the noisy quantum state $\rho^{\Escr}(\thetabf)$ as a convex combination of the noiseless ideal state $\Psi(\thetabf)$ and of the error density matrix $\tilde{\rho}(\thetabf)$, defined in \eqref{eq:errordensity}. With the definition \eqref{eq:errordensity}, the following set of relation holds
for any $y \in \{1,\hdots,N_h\}$ and $\thetabf \in \Real^D$:
\begin{align*}
    &\nu(p^{\Escr}(y|\thetabf)) \non \\&= ((1-\gamma) p(y|\thetabf)+\gamma\tilde{p}(y|\theta))(1-(1-\gamma) p(y|\thetabf)-\gamma\tilde{p}(y|\thetabf))\\
    &=(1-\gamma) p(y|\thetabf)\Bigl(1-(1-\gamma) p(y|\thetabf) \Bigr)\non \\&\quad + \gamma\tilde{p}(y|\thetabf) \Bigl(1-(1-\gamma) p(y|\theta) \Bigr)\\&\quad -\gamma (1-\gamma)p(y|\theta)\tilde{p}(y|\thetabf)-\gamma^2\tilde{p}(y|\thetabf)^2 \nonumber \\
    &=(1-\gamma)\nu(p(y|\thetabf))+ \gamma\biggl[(1-\gamma) p^2(y|\thetabf)\non \\& \quad + \tilde{p}(y|\thetabf) \Bigl(1-(1-\gamma) p(y|\thetabf) \Bigr)-\\& \quad (1-\gamma) p(y|\thetabf)\tilde{p}(y|\thetabf)-\gamma \tilde{p}(y|\thetabf)^2 \biggr]\nonumber\\
    &=(1-\gamma) \nu(p(y|\thetabf))+ \gamma\Bigl[ (1-\gamma) (p(y|\thetabf)-\tilde{p}(y|\thetabf))^2\\&\quad +\tilde{p}(y|\thetabf)-\tilde{p}(y|\thetabf)^2\Bigr]\\
    &=(1-\gamma) \nu(p(y|\thetabf))+\nu((1-\gamma)) (p(y|\thetabf)-\tilde{p}(y|\thetabf))^2\\ & \quad + \gamma \nu(\tilde{p}(y|\thetabf)) .
\end{align*}

We now analyze the bias term. By \eqref{eq:noisy_decomposition}, we have the equality
\begin{align}
   & \lVert \mathrm{bias}\rVert^2\non \\&=\sum_{d=1}^D \frac{1}{4} \biggl(\Tr\biggl(H\Bigl(\rho^{\Escr}\Bigl(\thetabf+{e}_d \frac{\pi}{2}\Bigr)-\rho^{\Escr}\Bigl(\thetabf-{e}_d\frac{\pi}{2}\Bigr)\non\\&-\Psi\Bigl(\thetabf+{e}_d \frac{\pi}{2}\Bigr)+\Psi\Bigl(\thetabf-{e}_d \frac{\pi}{2}\Bigr)\Bigr)\biggr)^2\\
    & \stackrel{(a)}{\leq} \sum_{d=1}^{D} \frac{1}{4} \lVert H\rVert^2_{\infty} \biggl( \biggl \lVert \rho^{\Escr}\Bigl(\thetabf+{e}_d \frac{\pi}{2}\Bigr)-\Psi\Bigl(\thetabf+{e}_d\frac{\pi}{2}\Bigr)\biggr \rVert_1\non \\&+\biggl \lVert \rho^{\Escr}\Bigl(\thetabf-{e}_d\frac{\pi}{2}\Bigr)-\Psi\Bigl(\thetabf-{e}_d\frac{\pi}{2}\Bigr)\biggr \rVert_1\biggr)^2\\
    & \stackrel{(b)}{\leq}4D \lVert H \rVert ^2_{\infty}\gamma, \label{eq:66}
\end{align}where the inequality in $(a)$ follows from tracial matrix H{\"o}lder's inequality and the triangle inequality for trace norm \cite{wilde2013quantum}. To obtain the inequality in $(b)$, we use the following relationship between trace distance and quantum fidelity $F(\rho,\sigma)=\Bigl(\Tr(\sqrt{\rho}\sigma \sqrt{\rho}) \Bigr)^2$ \cite{wilde2013quantum}
\begin{align}
  &\biggl\lVert\rho^{\Escr}\Bigl(\thetabf\pm{e}_j\frac{\pi}{2}\Bigr)-\Psi\Bigl(\thetabf\pm{e}_j\frac{\pi}{2}\Bigr)\biggr\rVert_1 \non \\&\leq 2 \sqrt{1-F\biggl(\rho^{\Escr}\Bigl(\thetabf\pm {e}_j\frac{\pi}{2}\Bigr),\Psi\Bigl(\thetabf\pm {e}_j\frac{\pi}{2}\Bigr)\biggr)}. \label{eq:33}
\end{align} Furthermore, the  quantum fidelity in \eqref{eq:33} can be upper bounded using \eqref{eq:decomposition_density} as
\begin{align}
& F(\rho^{\Escr}(\thetabf),\Psi(\thetabf))\non \\&=\langle \Psi(\thetabf) \vert \rho^{\Escr}(\thetabf) \vert \Psi(\thetabf) \rangle \non  \\
 &=(1-\gamma)\langle \Psi(\thetabf) \vert \Psi(\theta) \vert \Psi(\thetabf) \rangle+\gamma\langle \Psi(\thetabf) \vert \tilde{\rho}(\thetabf) \vert \Psi(\theta) \rangle \non \\
 &=(1-\gamma)+\gamma\langle \Psi(\thetabf) \vert \tilde{\rho}(\thetabf) \vert \Psi(\thetabf) \rangle \label{eq:55}\\
 & \geq 1-\gamma \label{eq:56}
\end{align} for all $\theta \in \Real^D$. Here, the last inequality follows since the density operator is positive semi-definite, making the second term in \eqref{eq:55} non-negative. Using \eqref{eq:56} in \eqref{eq:33} yields the inequality in \eqref{eq:66}.
\section{Additional Properties }\label{app:properties}
In this section, we  provide the following additional properties of the variance term $\nu(p^{\Escr}(y|\thetabf))$ defined in \eqref{eq:relation_1}. All these properties can be easily verified using the definition \eqref{eq:relation_1} and hence we omit the detailed derivations here.
\begin{enumerate}
    \item The inequality $\nu(p^{\Escr}(y|\thetabf)) \geq (1-\gamma) \nu(p(y|\thetabf)) +\gamma \nu(\tilde{p}(y|\thetabf))$ holds for all $y \in \{1,\hdots, N_h\}$ and $\theta \in \Theta$.
    \item $\nu(p^{\Escr}(y|\thetabf))$ is a concave function of $\gamma$. It is increasing in the range $\gamma \in [0,\gamma^{*}(y,\thetabf))$, and decreasing in the range $\gamma \in [\gamma^{*}(y,\thetabf),1)$, where
    \begin{align}
        \gamma^{*}(y,\thetabf)= \min \biggl \lbrace 1, 0.5 \biggl( 1- \frac{\nu({p}(y|\thetabf))-\nu(\tilde{p}(y|\thetabf))}{(\tilde{p}(y|\thetabf)-p(y|\thetabf))^2}\biggr)\biggr \rbrace.
    \end{align}
    \end{enumerate}
\section{Proof of Theorem~5.1}\label{app:variance_QEM}
In this section, we first compute the variance of the QEM-based stochastic gradient estimator \eqref{eq:QEM_gradientestimator}. With the notation $\thetabf^t_{d,\pm}=\thetabf^t\pm{e}_d \frac{\pi}{2}$, we have the equalities
\begin{align}
    \mathrm{var}(\xi^{\QEM}_t)&=\Ebb[\lVert\hat{g}^{\QEM}_t -g_t\rVert^2]\\
    &=\Ebb[\lVert \hat{g}^{\QEM}_t-g_t^{\Circ}+ g^{\Circ}_t-g_t\rVert^2] \\
    & \stackrel{(a)}{=} \underbrace{\Ebb[\lVert \hat{g}^{\QEM}_t-g_t^{\Circ}\rVert^2]}_{A}+ \underbrace{\Ebb[\lVert g^{\Circ}_t-g_t\rVert^2 ]}_{B}\label{eq:QEMapp_1},
    \end{align} where the expectation is over measurements as well as over sampled noisy circuits, and  we used the definition\begin{align}
 [g^{\Circ}_t]_d= &\frac{Z}{2 N_c}\sum_{l=1}^{N_c}\sgn(q_{\sbfbar_l}) \Bigl({\langle H \rangle}_{\rho_{\sbfbar_l}(\thetabf^t_{d,+})} \non \\&-{\langle H \rangle}_{\rho_{\sbfbar_l}(\thetabf^t_{d,-})} \Bigr) \label{eq:auxiliary_loss_app}.
 \end{align} The equality in $(a)$ follows from the equality $\Ebb[( \hat{g}^{\QEM}_t-g_t^{\Circ})^T (g^{\Circ}_t-g_t)|\sbfbar_{1:N_c}]=0 $ with the expectation taken over the quantum measurements conditioned on the sampled noisy circuit indices $\sbfbar_{1:N_c}$. We now analyze each of the terms in \eqref{eq:QEMapp_1} separately. The first term $A$ in \eqref{eq:QEMapp_1} evaluates as
\begin{align}
    A&=\frac{Z^2}{4N_c^2}\sum_{d=1}^D\Ebb\biggl[\biggl(\sum_{l=1}^{N_c} \sgn(q_{\sbfbar_l}) \Bigl(\underbrace{\widehat{\langle H \rangle}_{\rho_{\sbfbar_l}(\thetabf^t_{d,+})}-\langle H \rangle_{\rho_{\sbfbar_l}(\thetabf^t_{d,+})}}_{A_{1,l}}\non \\&+\underbrace{\langle H \rangle_{\rho_{\sbfbar_l}(\thetabf^t_{d,-})}-\widehat{\langle H \rangle}_{\rho_{\sbfbar_l}(\thetabf^t_{d,-})}}_{A_{2,l}} \Bigr) \biggr)^2 \biggr]\nonumber\\
    &\stackrel{(a)}{=}\frac{Z^2}{4N_c^2} \sum_{d=1}^D\mathrm{var}\biggl(\sum_{l=1}^{N_c} \sgn(q_{\sbfbar_l})(A_{1,l}+A_{2,l}) \biggr)\non\\
    &\stackrel{(b)}{=}\frac{Z^2}{4N_c^2}\sum_{d=1}^D\sum_{l=1}^{N_c}\Ebb[A_{1,l}^2+A_{2,l}^2] \label{eq:A_intermediate},
\end{align}where the equality in $(a)$ follows by noting that for the random variable $C=\sum_{l=1}^{N_c} C_l$, with $C_l=\sgn(q_{\sbfbar_l})(A_{1,l}+A_{2,l})$, we have the equality  $\Ebb[C|\sbfbar_{1:N_c}]=0$, where the expectation is over quantum measurements conditioned on the indices $\sbfbar_{1:N_c}$ of the sampled circuits. The equality in $(b)$ follows since  random variables $C_l$ for $l=1,\hdots,N_c$, are independent, enabling the equality $\var(C)=\sum_{l=1}^{N_c} \var(C_l)$. Furthermore, we have $\var(C_l)=\Ebb[(A_{1,l}+A_{2,l})^2]$, which  is in turn equal to $\Ebb[A_{1,l}^2+A_{2,l}^2]$, since, conditioned on a sampled noisy circuit, the expectation $\Ebb[A_{1,l}A_{2,l}|\sbfbar_l]=0$ holds.

Now, following an analysis similar to \eqref{eq:acrossproofs_1}, defining the Bernoulli random variable $W_{y,\sbfbar_l}=\{Y^{\QEM}_{\sbfbar_l}=y\}$, with $Y^{\QEM}_{\sbfbar_l}$ being the random variable that denotes the index of the measurement of the observable $H$, we obtain the inequality
\begin{align}
 &\Ebb[A_{1,l}^2|\sbfbar_{l}]\non\\& = \frac{1}{N_m/N_c}\Ebb\biggl[\biggl(\sum_{y=1}^{N_h} h_y\Bigl(W_{y,\sbfbar_l}-p_{\sbfbar_l}(y|\theta^t_{d,+})\Bigr) \biggr)^2 \Bigl |\sbfbar_{l} \biggr]\label{eq:A11} \\
 & \leq \frac{ \Tr(H^2)}{N_m/N_c}\sum_{y=1}^{N_h} \nu\Bigl(p_{\sbfbar_l}(y|\theta^t_{d,+}) \Bigr). \label{eq:A11_upperbound}
\end{align}A similar bound can be obtained for the term $\Ebb[A_{2,l}^2|\sbfbar_{1:N_c}]$. 

Using \eqref{eq:A11_upperbound}, we first obtain an upper bound on the term $A$ in \eqref{eq:A_intermediate} as
\begin{align}
    A \leq &\frac{Z^2}{4N_c^2} \frac{\Tr(H^2)}{N_m/N_c} \sum_{d=1}^D \sum_{l=1}^{N_c}\biggl(\sum_{y=1}^{N_h}\Ebb_{\sbfbar_l}\biggl[\nu\Bigl(p_{\sbfbar_l}(y|\theta^t_{d,+}) \Bigr)\biggr]\non \\& + \sum_{y=1}^{N_h}\Ebb_{\sbfbar_l}\biggl[\nu\Bigl(p_{\sbfbar_l}\Bigl(y|\theta^t_{d,-}) \Bigr)\biggr]\biggr) \\
     \leq &\frac{N_h D Z^2 \Tr(H^2)}{2N_m}  \sup_{\theta \in \Real^D, y \in \{1,\hdots,N_h\}}\Ebb_{\sbfbar}\Bigl[\nu\Bigl(p_{\sbfbar}(y|\theta) \Bigr)\Bigr], \label{eq:Aupperbound}
\end{align}
where we used the notation $\Ebb_{\bullet}$ to represent the expectation over the random variable in the subscript ${\bullet}$.

To get a lower bound on the term $A$, we start by observing that
\begin{align}
    \Ebb[A^2_{1,l}]&=\Ebb_{\sbfbar_{l}}[\Ebb[A^2_{1,l}|\sbfbar_{l}]]\label{eq:QEM_2} \\
    & \stackrel{(a)}{\geq} p_{\sbfbar=\{0\}^D} \Ebb[A^2_{1,l}|\sbfbar=\{0\}^D]\\
    &=p_{\sbfbar=\{0\}^D}N_c \var (\widehat{\langle H \rangle}_{\rho^{\Escr}(\theta^t_{d,+})}),
\end{align} with the inequality in $(a)$ obtained by lower bounding the outer expectation in \eqref{eq:QEM_2} with the choice of the $\sbfbar=\{0\}^D$th circuit. This determines the noisy operation $\Oscr^{\theta}_{\sbfbar}(\cdot)$ with Pauli  operator $\Pscr_0$ acting on all gates times the probability $p_{\sbfbar=\{0\}^D}$. With this choice of sampled circuit, we then have that $p_{\sbfbar_l}(y|\thetabf^t_{d,\pm} )=p^{\Escr}(y|\thetabf^t_{d,\pm})$ and $W_{y,\sbfbar_l}=\Ibb\{Y^{\Escr}=y\}$. This results in the last equality with $\var (\widehat{\langle H \rangle}_{\rho^{\Escr}(\theta^t_{d,+})})$ defined as in \eqref{eq:QEM_5}. Similar bound holds for $\Ebb[A^2_{2,l}]$. Subsequently, one obtains the lower bound
\begin{align}
    A \geq &\frac{Z^2}{4N_c}p_{\sbfbar=\{0\}^D}\sum_{d=1}^D \biggl[N_c \var (\widehat{\langle H \rangle}_{\rho^{\Escr}(\theta^t_{d,+})})\non \\&+N_c \var (\widehat{\langle H \rangle}_{\rho^{\Escr}(\theta^t_{d,-})}) \biggr]\\
    = &Z^2 p_{\sbfbar=\{0\}^D} \mathrm{\var}(\xi^{\Escr}_t)=c_1(\gamma)\mathrm{\var}(\xi^{\Escr}_t), \label{eq:Alowerbound}
\end{align}with $\mathrm{\var}(\xi^{\Escr}_t)$ defined as in \eqref{eq:2}.


We now evaluate the second term $B$ in \eqref{eq:QEMapp_1}. Towards this, we note the inequality
\begin{align}
    B&=\sum_{d=1}^D \Ebb[([g^{\Circ}_t]_d-[g_t]_d)^2]\non \\&= \frac{Z^2}{4N_c^2} \sum_{d=1}^D\sum_{l=1}^{N_c} \mathrm{var}(\sgn(q_{\sbfbar_l})[\langle H \rangle_{\rho_{\sbfbar_l}(\thetabf^t_{d,+})}-\langle H \rangle_{\rho_{\sbfbar_l}(\thetabf^t_{d,-})}])\non\\
    & \stackrel{(a)}{\leq }\frac{Z^2}{4N_c^2} \sum_{d=1}^D\sum_{l=1}^{N_c} \Ebb\biggl[\biggl(\sgn(q_{\sbfbar_l})[\langle H \rangle_{\rho_{\sbfbar_l}(\thetabf^t_{d,+})}\\ &\hspace{0.5cm}-\langle H \rangle_{\rho_{\sbfbar_l}(\thetabf^t_{d,-})} ]\biggr)^2\biggr]\non
    \end{align} \begin{align}
    &=\frac{Z^2}{4N_c} \sum_{d=1}^D \Ebb_{\sbfbar}\biggl[\Bigl(\langle H \rangle_{\rho_{\sbfbar}(\thetabf^t_{d,+})}-\langle H \rangle_{\rho_{\sbfbar}(\thetabf^t_{d,-})} \Bigr)^2\biggr]\non \\
    &=\frac{Z^2}{4N_c} \sum_{d=1}^D \Ebb_{\sbfbar}\biggl[\Bigl(\Tr\Bigl(H \Bigl(\rho_{\sbfbar}(\thetabf^t_{d,+})-\rho_{\sbfbar}(\thetabf^t_{d,-})\Bigr) \Bigr)\Bigr)^2\biggr]\non \\
    & \stackrel{(b)}{\leq} \frac{Z^2 \lVert H \rVert_{\infty}^2}{4N_c}\sum_{d=1}^D  \Ebb_{\sbfbar}\biggl[\lVert \rho_{\sbfbar}(\thetabf^t_{d,+}) -\rho_{\sbfbar}(\thetabf^t_{d,-})\rVert^2_1\biggr]\non\\
    & \stackrel{(c)}{\leq} \frac{D Z^2 \lVert H \rVert_{\infty}^2}{N_c} 
    \label{eq:Bupperbound},
\end{align} where in $(a)$ we used $\var(X) \leq \Ebb[X^2]$. The inequality in $(b)$ follows by the application of the tracial matrix H{\"o}lder's inequality \cite{baumgartner2011inequality}, and the inequality $(c)$ follows since trace norm between density matrices is upper bounded by 2 \cite{wilde2013quantum}. 

Noting the inequality $B\geq 0$ and using \eqref{eq:Alowerbound}, we get the following lower bound
\begin{align*}
    \var(\xi^{\QEM}_t) \geq c_1(\gamma) \var(\xi^{\Escr}_t).
\end{align*} Furthermore, using \eqref{eq:Aupperbound} and \eqref{eq:Bupperbound} in \eqref{eq:QEMapp_1} yields the upper bound $V^{\QEM}$ of \eqref{eq:VQEM}.

To show the relation in \eqref{eq:relation_3}, we start with the upper bound on term $A$ obtained in \eqref{eq:Aupperbound}. One can lower bound this term as
\begin{align}
   & \frac{N_hD Z^2\Tr(H^2)}{2N_m}\sup_{\theta \in \Real^D, y \in \{1,\hdots,N_h\}}\Ebb_{\sbfbar}[\nu(p_{\sbfbar}(y|\thetabf))] \non \\
   & \geq \frac{N_h D Z^2\Tr(H^2)p_{\sbfbar=\{0\}^D}}{2N_m}\sup_{\theta \in \Real^D, y \in \{1,\hdots,N_h\}}\nu(p^{\Escr}(y|\thetabf)) \non \\
   &= Z^2p_{\sbfbar=\{0\}^D} V^{\Escr} =c_1(\gamma)V^{\Escr},
\end{align} where $V^{\Escr}$ is as defined in \eqref{eq:variance_noisy}. Together with the bound in \eqref{eq:Bupperbound}, we then get the inequality
\begin{align}
    V^{\QEM} \geq c_1(\gamma)V^{\Escr}+c_2(\gamma) \frac{D \lVert H \rVert_{\infty}^2 }{N_c}
\end{align}with $c_2(\gamma)=Z^2$.

Finally, we analyze the scalar functions $c_1(\gamma)$ and $c_2(\gamma)$ when the quantum noise channel in \eqref{eq:noisemodels} is depolarizing. In this case,  \cite[Thm. 2]{takagi2021optimal} gives the equalities
\begin{align}
    Z_d = \frac{1+\bigl( 1-2^{1-2n}\bigr)\Bigl( 1-(1-\gamma)^{1/D}\Bigr)}{(1-\gamma)^{1/D}}
\end{align} and
\begin{align}
    p_d(0) = \frac{2^{2n}-1+(1-\gamma)^{1/D}}{2^{2n}\Bigl(1+(1-2^{1-2n})(1-(1-\gamma)^{1/D}) \Bigr)}
\end{align} for $d=1,\hdots,D$. Note that the inequalities $Z_d\geq 1$ and $\frac{dZ_d}{d \gamma} \geq 0$ holds, whereby $Z_d$ is a non-decreasing function of $\gamma$. Subsequently, we have the equality $c_2(\gamma)=Z^2=(Z_d)^{2D} \geq 1$ and $c_2(\gamma)$ is a non-decreasing function of $\gamma$.


We now analyze the term $c_1(\gamma)=Z^2p_{\sbfbar=\{0\}^D}=(Z_d^2p_{d,s=0})^D$, where we have
$
 Z_d^2p_{d,s=0}= \frac{\mathrm{num}}{\mathrm{den}}$ with $\mathrm{num}= (2^{2n+1}-4+2^{1-2n})+(1-\gamma)^{1/D}(5-4/2^{2n}-2^{2n})-(1-\gamma)^{2/D}(1-2^{1-2n})$ and $\mathrm{den}=2^{2n}(1-\gamma)^{2/D}$. The following two properties can be easily verified.
 Firstly, we show that
 \begin{align*}
&\mathrm{num}-\mathrm{den}\non \\&= (2^{2n+1}-4+2^{1-2n})+(1-\gamma)^{1/D}(5-4/2^{2n}-2^{2n})\non \\&\quad-(1-\gamma)^{2/D}(1-2^{1-2n}+2^{2n}) \\
& \stackrel{(a)}\geq (2^{2n+1}-4+2^{1-2n})+(1-\gamma)^{1/D}(5-4/2^{2n}-2^{2n}\non \\&\quad-1+2^{1-2n}-2^{2n})\\
&=(2^{2n+1}-4+2^{1-2n})-(1-\gamma)^{1/D}(-4+2/2^{2n}+2^{2n+1})\\
& \geq 0,
 \end{align*}where inequality $(a)$ follows since $(1-\gamma)^2 \leq (1-\gamma)$ and the last inequality follows since $(1-\gamma)^{1/D}\leq 1$. As a consequence, we have the equality
 $Z_d^2p_{d,s=0}=\mathrm{num}/\mathrm{den} \geq 1$, and thus the inequality $c_1(\gamma)\geq 1$. 
 
 Lastly, it can be verified that
 \begin{align}\frac{d (Z_d^2p_{d,s=0})}{d\gamma}&=\frac{d}{d\gamma} \frac{\mathrm{num}}{\mathrm{den}}=\frac{1}{D(1-\gamma)^{2/D+1}}\Bigl(4-\frac{8}{2^{2n}}+\frac{4}{2^{4n}}\Bigr) \non \\&+\frac{1}{D(1-\gamma)^{1/D+1}}\Bigl(\frac{5}{2^{2n}}-\frac{4}{2^{4n}}-1 \Bigr)\non \\
 & \geq \frac{1}{D(1-\gamma)^{1/D+1}}\Bigl(3-\frac{3}{2^{2n}}\Bigr) \geq 0 \end{align} holds, where the last inequality follows by the inequality $(1-\gamma)^2 \leq (1-\gamma)$.
 Consequently, $Z_d^2p_{d,s=0}$, and thus $c_1(\gamma)$, are non-decreasing functions of $\gamma$.
\section{Additional Experiments}
In Figure~\ref{fig:low_error} of Section~\ref{sec:experiments}, we have evaluated the performance of the parameter iterate $\thetabf^t$ of the SGD  in terms of the loss function $L(\thetabf)$ in \eqref{eq:problem1} which can be estimated in practice given an infinite  measurements on a gate-noise free PQC  during testing. In this section, we assume that during testing, a gate-noise free PQC with the same measurement budget as that during training is available. 
As can be seen from Fig.~\ref{fig:low_error_sameshot}, obtained under the same conditions as Fig.~\ref{fig:low_error}, the finite number of measurements of the PQC made during testing induces a larger variance on the evaluated loss function for all settings.
\begin{figure*}[h!]
     \begin{minipage}[c]{0.45 \linewidth}
     \centering
      \includegraphics[scale=0.5,clip=true, trim=0in 0in 0in 0.45in]{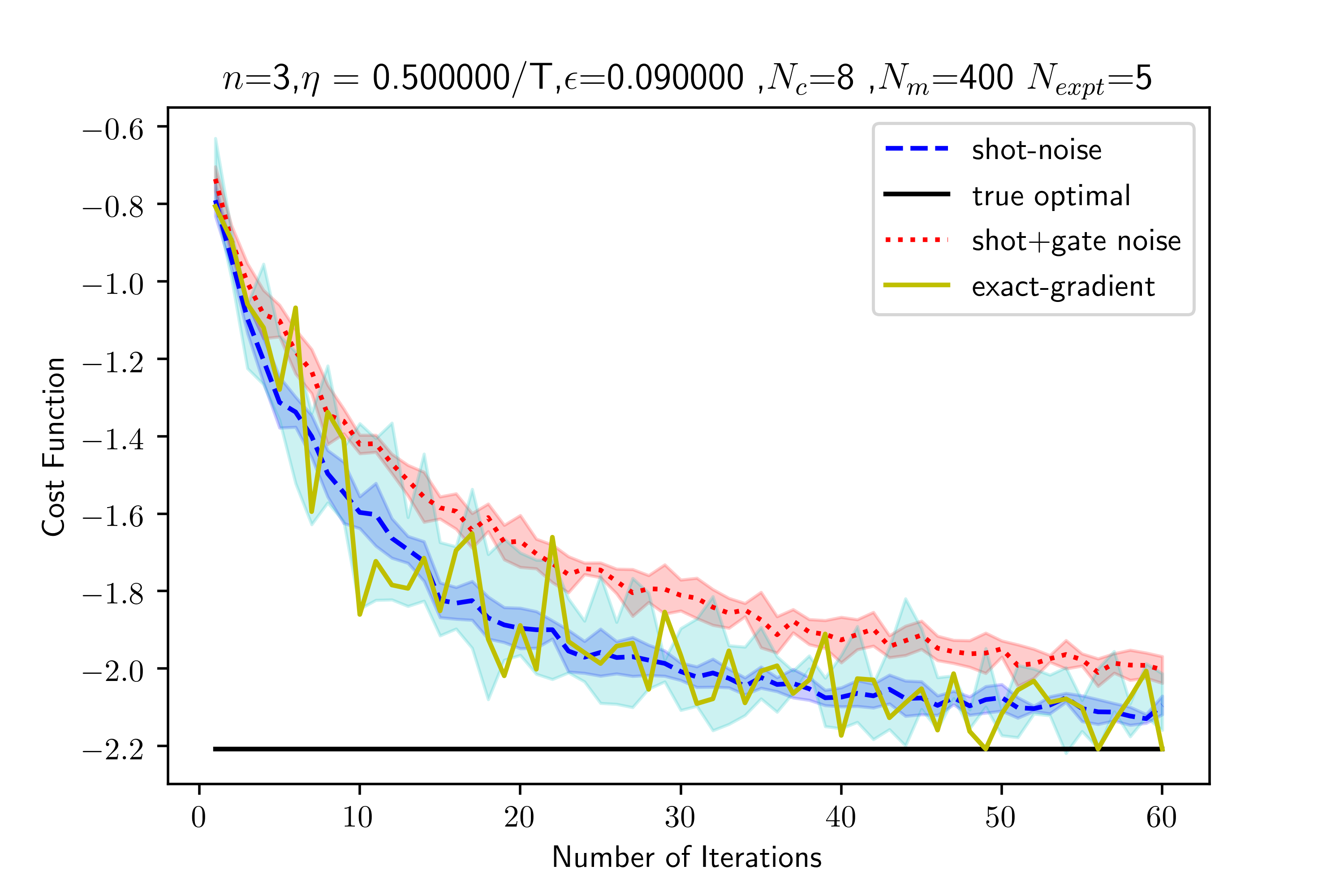}
      \end{minipage} \hfill
      \begin{minipage}[c]{0.45 \linewidth}
      \centering
     \includegraphics[scale=0.5,clip=true, trim=0in 0in 0in 0.45in]{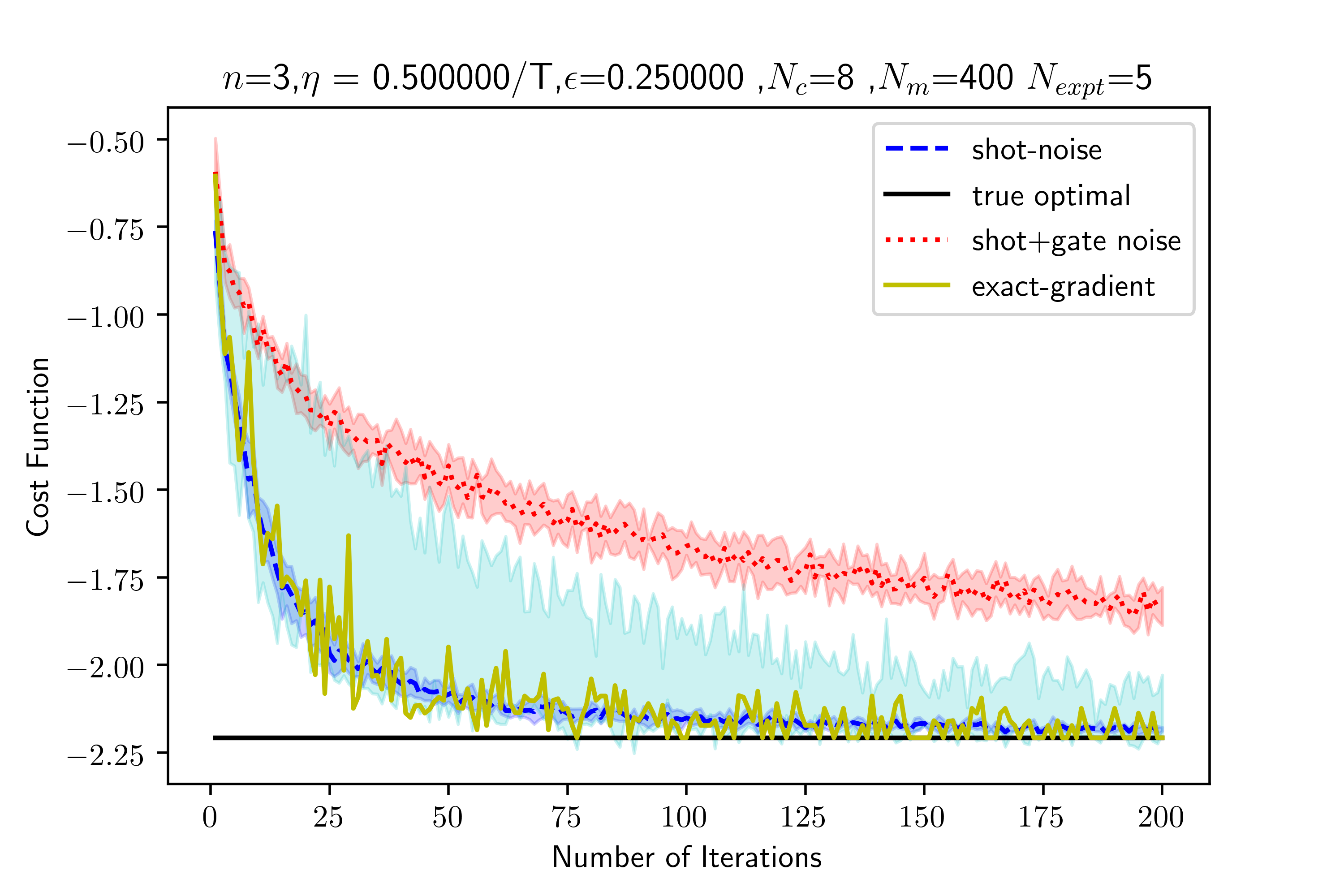}
     \end{minipage}
         \caption{Convergence analysis of the SGD as a function of the number of iterations when exact gradients can be computed; when only shot-noise is present; when shot and gate noise are present; and when QEM is employed - ($n=3$ qubit system subjected to depolarizing noise on the CNOT gates). The learning rate is set as $\eta_t = 0.5/t$; $N_c=8$; and $N_m=400$. (Left) low noise with $\epsilon=0.09$; (right) large noise with $\epsilon=0.25$. The loss function $L(\thetabf^t)$ is evaluated during testing by taking $N_m$ measurements of a gate-noise free PQC.}
         \label{fig:low_error_sameshot}
\end{figure*}
\bibliography{ref}
\bibliographystyle{IEEEtran}
\end{document}

%% file: macros-ieee.tex
\usepackage{amsmath, amssymb, bbm, xspace}
\usepackage{epsfig}
\usepackage{longtable}
\usepackage{color}
\usepackage{mathrsfs}
\usepackage{comment}

\usepackage{courier}

\newtheorem{theorem}{Theorem}[section]
\newtheorem{assumption}{Assumption}[section]

\newtheorem{lemma}{Lemma}[section]

\def\bkE{{\rm I\kern-.17em E}}
\def\bk1{{\rm 1\kern-.17em l}}
\def\bkD{{\rm I\kern-.17em D}}
\def\bkR{{\rm I\kern-.17em R}}
\def\bkP{{\rm I\kern-.17em P}}

\def\bkZ{{\bf{Z}}}

\def\bkE{{\rm I\kern-.17em E}}
\def\bk1{{\rm 1\kern-.17em l}}
\def\bkD{{\rm I\kern-.17em D}}
\def\bkR{{\rm I\kern-.17em R}}
\def\bkP{{\rm I\kern-.17em P}}

\makeatletter
\newcommand{\pushright}[1]{\ifmeasuring@#1\else\omit\hfill$\displaystyle#1$\fi\ignorespaces}
\newcommand{\pushleft}[1]{\ifmeasuring@#1\else\omit$\displaystyle#1$\hfill\fi\ignorespaces}
\makeatother

\def\bkZ{{\bf{Z}}}
\def\b12{(\beta_1,\beta_2)}

\newcounter{example}
\renewcommand{\theexample}{\thesection.\arabic{example}}

\newcounter{remark}
\renewcommand{\theremark}{\thesection.\arabic{remark}}

\def\Ebb{\mathbb{E}}
\newlength{\noteWidth}
\long\def\notes#1{\ifinner
{\tiny #1}
\else
\marginpar{\parbox[t]{\noteWidth}{\raggedright\tiny #1}}
\fi\typeout{#1}}

 \def\notes#1{\typeout{read notes: #1}} 

\newcommand{\ie}{i.e.\@\xspace} 
\newcommand{\etal}{et al.\@\xspace} 

\newcommand{\Real}{\ensuremath{\mathbb{R}}}

\def\Ebb{\mathbb{E}}

\def\Ibb{{\mathbb{I}}}

\def\ast{\buildrel{t}\over\rightarrow}

\def\exp{\mathop{\hbox{\rm exp}}}

\def\sgn{\mathop{\hbox{\rm sgn}}}

\def\spose#1{\hbox to 0pt{#1\hss}}

\def\text #1{\hbox{\quad#1\quad}}

\def\Escr{\mathcal{E}}

\def\nthinsp{\mskip -2   mu}

\def\superstar{^{\raise 0.5pt\hbox{$\nthinsp *$}}}
\def\SUPERSTAR{^{\raise 0.5pt\hbox{$*$}}}

\def\lamstarT {\lambda^{\raise 0.5pt\hbox{$\nthinsp *$}T}}

\def\sbar{\skew3\bar s}

\def\Dscr{{\cal D}}

\def\Lscr{{\cal L}}

\def\Oscr{{\cal O}}
\def\Pscr{{\cal P}}

\def\Uscr{{\cal U}}

\def\Gscr{{\cal G}}

\def\sbar{\bar s}

\def\non{\nonumber}

\let\forallnew\forall
\renewcommand{\forall}{\forallnew\ }
\let\forall\forallnew

		\def\bkE{{\rm I\kern-.17em E}}
		\def\bk1{{\rm 1\kern-.17em l}}
		\def\bkD{{\rm I\kern-.17em D}}
		\def\bkR{{\rm I\kern-.17em R}}
		\def\bkP{{\rm I\kern-.17em P}}
		\def\bkY{{\bf \kern-.17em Y}}
		\def\bkZ{{\bf \kern-.17em Z}}
		\def\bkC{{\bf  \kern-.17em C}}

{\begin{list}{}%
         {\setlength{\leftmargin}{#1}}%
         \item[]%
}
{\end{list}}

		\def\bsp{\begin{split}}
		\def\beq{\begin{eqnarray}}
		\def\bal{\begin{align*}}
		\def\bc{\begin{center}}
		\def\be{\begin{enumerate}}
		\def\bi{\begin{itemize}}
		\def\bs{\begin{small}}
		\def\bS{\begin{slide}}
		\def\ec{\end{center}}
		\def\ee{\end{enumerate}}
		\def\ei{\end{itemize}}
		\def\es{\end{small}}
		\def\eS{\end{slide}}
		\def\eeq{\end{eqnarray}}
		\def\eal{\end{align*}}
		\def\esp{\end{split}}
		\def\qed{ \vrule height7.5pt width7.5pt depth0pt}  

	\def\cp2problem#1#2#3#4{\fbox
		 {\begin{tabular*}{0.9\textwidth}
			{@{}l@{\extracolsep{\fill}}l@{\extracolsep{6pt}}l@{\extracolsep{\fill}}c@{}}
				#1 & & $#4 $ 
			\end{tabular*}}}

		\def\bkE{{\rm I\kern-.17em E}}
		\def\bk1{{\rm 1\kern-.17em l}}
		\def\bkD{{\rm I\kern-.17em D}}
		\def\bkR{{\rm I\kern-.17em R}}
		\def\bkP{{\rm I\kern-.17em P}}
		
		\def\bkZ{{\bf{Z}}}

\newcommand {\beeq}[1]{\begin{equation}\label{#1}}
\newcommand {\eeeq}{\end{equation}}
\newcommand {\bea}{\begin{eqnarray}}
\newcommand {\eea}{\end{eqnarray}}

\def\texitem#1{\par\smallskip\noindent\hangindent 25pt
               \hbox to 25pt {\hss #1 ~}\ignorespaces}

\def\bsp{\begin{split}}
		\def\beq{\begin{eqnarray}}
		\def\bal{\begin{align*}}
		\def\bc{\begin{center}}
		\def\be{\begin{enumerate}}
		\def\bi{\begin{itemize}}
		\def\bs{\begin{small}}
		\def\bS{\begin{slide}}
		\def\ec{\end{center}}
		\def\ee{\end{enumerate}}
		\def\ei{\end{itemize}}
		\def\es{\end{small}}
		\def\eS{\end{slide}}
		\def\eeq{\end{eqnarray}}
		\def\eal{\end{align*}}
		\def\esp{\end{split}}
		\def\qed{ \vrule height7.5pt width7.5pt depth0pt}  

\usepackage{amsmath, amssymb, xspace}
\usepackage{epsfig}
\usepackage{longtable}
\usepackage{color}
\usepackage{mathrsfs}
\usepackage{subfig}
\newenvironment{proof}[1][]{{\noindent \textit{ Proof}: }}{\hfill \qed \vspace{3pt}\\ }

